\documentclass[final,reqno]{siamltex}

\usepackage{amsmath,amssymb}
\usepackage{bm,cite,mathrsfs,nicefrac,upgreek}
\usepackage[dvips]{graphicx}
\usepackage{color}
\usepackage{sgame}
\newcommand{\cA}{{\mathcal{A}}}
\newcommand{\oA}{{\Bar{\mathcal{A}}}}
\newcommand{\cS}{{\mathcal{S}}}
\newcommand{\oS}{{\Bar{\mathcal{S}}}}
\newcommand{\Cc}{{\mathcal{C}}}
\newcommand{\cX}{{\mathcal{X}}}
\newcommand{\Pm}{{\mathcal{P}}}
\newcommand{\Bor}{{\mathcal{B}}}
\newcommand{\cI}{\mathcal{I}}
\newcommand{\df}{\triangleq}
\newcommand{\Prob}{{\mathbb{P}}} 
\newcommand{\Exp}{{\mathbb{E}}} 
\newcommand{\Ind}{{\bm{1}}} 
\newcommand{\sF}{{\mathfrak{F}}} 
\newcommand{\abs}[1]{\lvert#1\rvert}
\newcommand{\babs}[1]{\bigl\lvert#1\bigr\rvert}
\newcommand{\supnorm}[1]{{\lVert}#1{\rVert}_{\infty}}
\newcommand{\magn}[1]{\left\vert #1 \right\vert}
\newcommand{\D}{\mathrm{d}}
\newcommand{\NN}{\mathbb{N}}
\newtheorem{claim}{Claim}[section]
\newtheorem{remark}{Remark}[section]

\title{Aspiration Learning in Coordination Games
\thanks{This work was supported by
ONR project \#N00014-09-1-0751 and AFOSR project
\#FA9550-09-1-0538.}}

\author{Georgios C. Chasparis\thanks{Department of Automatic Control,
Lund University, 221 00-SE Lund,
Sweden (georgios.chasparis@control.lth.se).
http://www.control.lth.se/chasparis.}
\and Ari Arapostathis\thanks{Department of Electrical and Computer
Engineering, The University of Texas at Austin, 1 University Station,
Austin, TX 78712 (ari@mail.utexas.edu).
http://www.ece.utexas.edu/\~\relax ari.
This author's  work was supported in part by the Office of Naval Research
through the Electric Ship Research and Development Consortium.}
\and Jeff S. Shamma\thanks{School of Electrical and
Computer Engineering, Georgia Institute of Technology, Atlanta, GA
30332 (shamma@gatech.edu). http://www.prism.gatech.edu/\~\relax jshamma3.}}

\begin{document}
\maketitle

\begin{abstract}
We consider the problem of distributed convergence to efficient
outcomes in coordination games through dynamics based on aspiration learning.
Under aspiration learning, 
a player continues to play an action as long as the rewards received 
exceed a specified aspiration level. Here, the aspiration level is a 
fading memory average of past rewards, and these levels also are 
subject to occasional random perturbations. A player becomes dissatisfied 
whenever a received reward is less than the aspiration level, in which case 
the player experiments with a probability proportional to the degree of
dissatisfaction. 
Our first contribution is the characterization
of the asymptotic behavior of the induced Markov chain of the
iterated process in terms of an equivalent finite-state Markov chain. We
then characterize explicitly the behavior of the proposed aspiration
learning in a generalized version of coordination games,
examples of which include network formation and common-pool games.
In particular, we show that in generic coordination games the
frequency at which an efficient action profile is played
can be made arbitrarily large. Although convergence to efficient
outcomes is desirable, in several coordination games, such as
common-pool games, attainability of fair outcomes, i.e., sequences
of plays at which players experience highly rewarding returns with
the same frequency, might also be of special interest. To this end,
we demonstrate through analysis and simulations that aspiration
learning also establishes fair outcomes in all symmetric
coordination games, including common-pool games.
\end{abstract}


\markboth{G.~C.~Chasparis, A.~Arapostathis and J.~S.~Shamma}
{Aspiration Learning in Coordination Games}


\section{Introduction}

Distributed coordination is of particular interest in many
engineering systems. Two examples are distributed overlay routing
or network formation \cite{Chun04} and medium access control
\cite{Inaltekin05} in wireless communications. In either case, nodes
need to utilize their resources \emph{efficiently} so that a
desirable global objective is achieved. For example, in network formation, 
nodes need to choose their immediate links so that connectivity is achieved with a 
minimum possible communication cost, i.e., minimum number of links.
Similarly, in medium access control, users need to establish a
\emph{fair} scheduling of accessing a shared communication channel
so that collisions (i.e., situations at which two or more users
access the common resource) are avoided. In these scenarios, achieving 
\emph{coordination} in a distributed and adaptive
fashion to an \emph{efficient} 
outcome is of special interest.

The distributed yet coupled nature of these problems, combined with a 
desire for online adaptation, motivates using models based on game 
theoretic learning \cite{FudenbergLevine98,Young04,Sandholm10}.
In game theoretic learning, each agent is endowed with a set of 
actions and a utility/reward function that depends on that agent's 
and other agents' actions. Agents then learn which action to play
based only on their own previous experience of the game (actions
played and utilities received).
A major challenge in this setting is that explicit utility function 
optimization may be impractical. This may be due to inherent complexity 
(e.g., a large number of players or actions), or the lack of any closed 
form expression for the utility function. Rather, rewards can be measured 
online. In terms of game theoretic learning, this eliminates adaptation 
based on an ability to compute a ``best reply''.  Another obstacle to utility 
maximization is that from any agent's perspective, the environment includes 
other adapting agents, and hence is nonstationary. Consequently, actions 
that may have been effective in the past need not continue to be effective.

Motivated by these issues, this paper considers a form of distributed
learning dynamics known as \emph{aspiration learning}, where agents
``satisfice" rather than ``optimize". The
aspiration learning scheme is based on a simple principle of
``win-stay, lose-shift" \cite{Posch99}, according to which
a successful action is repeated while an unsuccessful action is
dropped. The success of an action is determined by a simple
comparison test of its performance with the player's desirable
return (\emph{aspiration level}). The aspiration level is updated
to incorporate prior experience into the agent's success criterion.
Through this learning scheme, agents \emph{learn} to play their
``best" action.

The history of aspiration learning schemes starts with the
pioneering work of \cite{Simon55}, where satisfaction seeking
behavior was used to explain social decision making.
A simple aspiration learning model is presented in \cite{Posch99},
where games of two players and two actions are considered, and decisions
are taken based on the ``win-stay, lose-shift'' rule. 
In the special case of two-player/two-action \emph{mutual 
interest games} and \emph{symmetric coordination games}, respectively, references, 
\cite{Pazgal95} and \cite{Kim99} show that the
payoff-dominant action profile is selected with probability close to one.
Similar are the results in
\cite{Karandikar98,ChoMatsui05}. However, contrary to
\cite{Pazgal95} and \cite{Kim99}, both models incorporate a small
perturbation either in the aspiration update \cite{Karandikar98} or
in the action update \cite{ChoMatsui05}. 

Recent research efforts on equilibrium selection in games
have focused on achieving distributed convergence to
\textit{Pareto-efficient} payoff profiles, i.e., payoff profiles at
which no action change can make a player better off while not
making some other player worse off. For example,
reference \cite{MardenYoungPao11} introduced an aspiration learning
algorithm that converges (in distribution) to action
profiles that maximize social welfare in multiple player games.
Some key characteristics of
this algorithm is that agents keep track of their most recent
satisfactory action and satisfactory payoff (benchmark action and payoff), and
they update their actions by following a ``win-stay lose-shift''
rule, where the aspiration level is defined as the benchmark payoff.
Convergence to the Pareto-efficient payoffs in two player games also has
been investigated by \cite{ArieliBabichenko11}. The learning
algorithm considered in \cite{ArieliBabichenko11} has two
distinctive features: a) agents commit on playing a series of actions
for a $k$-period interval, and b) agents make decisions according to
a ``win-stay lose-shift'' rule, where aspiration levels are computed
as the running average payoff over all the previous $k$-period
intervals. It is shown that, in two player games, the agents' payoffs
converge to a small neighborhood of the set of the Pareto-efficient
payoffs almost surely if $k$ is sufficiently large.

In this paper, we also focus on achieving convergence to efficient
payoff profiles (also part of the Pareto-efficient payoff profiles)
in coordination games of large number of players and actions. Agents
apply an aspiration learning scheme that is motivated by
\cite{Karandikar98}. Our goal is to a) characterize explicitly the
asymptotic behavior of the process for generic games of
multiple players and actions, and b) derive conditions under which
\emph{efficient} payoffs are selected in large coordination games. Our
main contribution is the characterization of the asymptotic behavior
of the induced Markov chain by means of the invariant distributions of
an equivalent finite-state Markov
chain, whenever the experimentation probability becomes
sufficiently small.
This equivalence simplifies the analysis of what would otherwise
be an infinite state Markov process.
These results extend prior analysis on this type of
aspiration learning schemes to games of multiple players and
actions. We also specialize the results for a class of games that
is a generalized version of so-called \emph{coordination games}. In
particular, we show that, in these games, the unique invariant
distribution of the equivalent finite-state Markov chain puts
arbitrarily large weight on the \emph{payoff-dominant} action
profiles if the step size of the aspiration-level update becomes
sufficiently small. We finally demonstrate the utility of the
learning scheme to network formation games, which is of independent
interest, since prior learning schemes on network formation are
primarily based on best-response dynamics, e.g., \cite{Bala00}.

While convergence to payoff-dominant action profiles in
coordination games is desirable, another desirable property is a
notion of fairness. 
In particular, for some coordination games where coincidence of
interests is not so strong, such as the \emph{Battle of the Sexes} (cf.,
\cite[Section~2.3]{Osborne94}), convergence to a single action
profile might not be \emph{fair} for all agents that would probably
rather be in a different action profile. Instead, an alternation
between several action profiles might be more desirable, usually 
described through distributions in the \textit{joint} action
space. An example of a class of such coordination games is 
so-called common-pool games, where multiple users need to coordinate
on utilizing a limited common resource. The proposed aspiration
learning algorithm also may provide a distributed and adaptive
approach for convergence to fair outcomes in such symmetric coordination
games, such as common-pool games. This property is of independent interest,
since it is relevant to several scenarios of distributed resource
allocation, such as medium access control in wireless communications
\cite{Inaltekin05}.

In comparison to prior and other current work, this paper develops 
(and corrects) the specific  model of aspiration learning in
\cite{Karandikar98} beyond two player games.
The paper goes on to derive specialized results for coordination 
games involving convergence to efficient action profiles and fairness
in symmetric games.
The results in \cite{MardenYoungPao11} use a simpler finite state model of
aspiration learning and are applicable to almost all games.
The results in \cite{MardenYoungPao11} establish convergence to efficient
action profiles, but as yet do not specify selection/fairness 
among these action profiles.
The model of \cite{ArieliBabichenko11} is more closely related to the present 
model, but with a different definition of aspiration levels and a different
mechanism to perturb aspirations.
The results of convergence to efficiency in \cite{ArieliBabichenko11}
extend beyond coordination games while requiring two player games and do
not specify fairness/selection among efficient profiles.

The remainder of the paper is organized as follows. Section~\ref{S2}
defines coordination games and presents two special cases of coordination
games, namely network formation and common-pool games.
Section~\ref{S3} presents the aspiration learning algorithm and its
convergence properties in games of multiple players and actions.
Section~\ref{S4} specializes the convergence analysis to
coordination games and establishes convergence to efficient
outcomes. It also demonstrates the results through simulations in
network formation games. Section~\ref{S5} extends the convergence
analysis to symmetric coordination games and establishes conditions
under which convergence to fair outcomes is also established.
Finally, Section~\ref{S6} presents concluding remarks.

\emph{Terminology:} We consider the standard setup of finite
strategic-form games. There is a finite set of \emph{agents} or \emph{players},
$\cI=\{1,2,\dotsc,n\}$, and each agent has a finite set of
actions, denoted by $\cA_{i}$. The set of action profiles is
the Cartesian product $\cA\df \cA_{1}\times\dotsb\times\cA_{n}$;
$\alpha_{i}\in\cA_{i}$ denotes an \textit{action} of agent $i$; and
$\alpha=(\alpha_{1},\dotsc,\alpha_{n})\in\cA$ denotes the \textit{action
profile} or \textit{joint action} of all agents.
The \emph{payoff/utility function} of
player $i$ is a mapping $u_{i}:\cA\rightarrow\mathbb{R}$.
A strategic-form game, denoted $\mathscr{G}$, consists of 
the sets $\mathcal{I}$, $\mathcal{A}$ and the preference relation induced
by the utility functions $u_i$, $i\in\mathcal{I}$. An
action profile $\alpha^*\in\cA$ is a \emph{(pure) Nash
equilibrium} if
\begin{equation}\label{E-Nash}
u_{i}(\alpha_{i}^*,\alpha_{-i}^*) \geq u_{i}(\alpha_{i}',\alpha_{-i}^*)
\end{equation}
for all $i\in\cI$ and $\alpha_{i}'\in\cA_{i}$, where
$-i$ denotes the complementary set $\cI\setminus\{i\}$.
We denote the set of pure Nash equilibria by $\cA^*$. In
case the inequality \eqref{E-Nash} is strict, the Nash
equilibrium is called a \emph{strict Nash equilibrium}. For the
remainder of the paper, the term ``Nash
equilibrium" always refers to a ``pure Nash equilibrium."

\section{Coordination Games}\label{S2}

\subsection{Definitions}

Before defining coordination games, we first need to define the
notion of \emph{better reply}:

\begin{definition}[Better reply]
The better reply of agent $i\in\cI$ to an action
profile $\alpha=(\alpha_{i},\alpha_{-i})\in\cA$ is a set
valued map $\mathrm{BR}_{i}:\cA\to2^{\cA_{i}}$ such that
for any $\alpha_{i}^*\in\mathrm{BR}_{i}(\alpha)$ we have
$u_{i}(\alpha_{i}^*,\alpha_{-i}) > u_{i}(\alpha_{i},\alpha_{-i})$.
\end{definition}

A \emph{coordination game} is defined as follows:

\begin{definition}[Coordination game]\label{D2.2}
A game of two or more agents is a
coordination game if there exists
$\oA\subset\cA$ such that the following
conditions are satisfied:
\begin{enumerate}
\item[\upshape{(a)}]
for any $\Bar{\alpha}\in\oA$ and
$\alpha\notin\oA$,
\begin{equation}\label{E-CG1}
u_{i}(\Bar{\alpha}) \geq u_{i}(\alpha)\quad \text{~for all~} i\in\cI\,,
\end{equation}
i.e., $\oA$ payoff-dominates $\cA\setminus\oA\,$;

\item[\upshape{(b)}]
for any
$\alpha\in\cA\setminus(\cA^*\cup\oA)$,
there exist $i\in\cI$ and action
$\alpha_{i}'\in \mathrm{BR}_{i}(\alpha)$ such that
\begin{equation}\label{E-CG2}
u_{j}(\alpha_{i}',\alpha_{-i})\geq u_{j}(\alpha_{i},\alpha_{-i})\quad
\text{~for all~} j\neq{i}\,;
\end{equation}

\item[\upshape{(c)}]
for any
$\alpha^*\in\cA^*\setminus\oA$
(if non-empty), there exist an action profile $\Tilde{\alpha}\in\cA$
and a sequence of distinct agents $j_{1},\dotsc,j_{n-1}\in\cI$, such that
\begin{equation*}
u_{i}\left(\Tilde{\alpha}_{j_{1}},\dotsc,\Tilde{\alpha}_{j_{\ell}},
\alpha^*_{-\{j_{1},\dotsc,j_{\ell}\}}\right)< u_{i}(\alpha^*)
\end{equation*}
for all $i\in\{j_{1},j_2,\dotsc,j_{\ell+1}\}$, $\ell=1,2,\dotsc,n-1$.

\end{enumerate}
\end{definition}

A \emph{strict coordination game} refers to a coordination
game with the inequality \eqref{E-CG1} being
strict.

The conditions of a coordination game establish a weak form of
``coincidence of interests'' and define a larger class of games than
the ones traditionally considered as coordination games, e.g.,
\cite{Vanderschraaf01,Lewis02}. For example, according to
\cite{Lewis02}, one of the conditions that a coordination game needs
to satisfy is that payoff differences among players at any action
profile are much smaller than payoff differences among different
action profiles.
This condition reflects a form of coincidence of interests.
Definition~\ref{D2.2}~(b) also
establishes a similar form of coincidence of interests, but weaker
in the sense that it holds for at least one direction of action change.


Note also that existence of Nash equilibria is not necessary for a
game to be a coordination game.
Furthermore, if $\cA^*\subset\oA$,
then Definition~\ref{D2.2} can be written solely
with respect to the desirable set of profiles $\oA$.
In that case,
Definition~\ref{D2.2}~(c) becomes vacuous since
$\cA^*\setminus\oA=\varnothing$.

A trivial example of a coordination game is the Stag-Hunt Game of
Table~\ref{Tb:SHG}.

\begin{table}[!ht]
\centering
\begin{game}{2}{2}
 & A & B\\
A &$4,4$ &$0,2$\\
B &$2,0$ &$3,3$\\
\end{game}\\[5pt]
\caption{The Stag-Hunt Game}\label{Tb:SHG}
\end{table}

First, there exists a payoff-dominant profile, namely $(A,A)$, that
can be identified as the desirable set $\oA$, and
satisfies Definition~\ref{D2.2}~(a).
Also, from any action profile outside
$\cA^*\cup\oA$, namely $(A,B)$ or
$(B,A)$, there is a better reply that improves the payoff for all
agents (i.e., Definition~\ref{D2.2}~(b) holds).
Lastly, for any Nash equilibrium
profile outside $\oA$, i.e., $(B,B)$, there is a
player (row or column) and an action which makes everyone worse off
(i.e., Definition~\ref{D2.2}~(c) holds).
Thus, the Stag-Hunt game satisfies all
the conditions of Definition~\ref{D2.2}.

Note finally that in some games, there might be multiple choices
for the selection of the desirable set $\oA$. For
example, in the Stag-Hunt game of Table~\ref{Tb:SHG}, an
alternative selection of $\oA$ corresponds to the
union of the action profiles $(A,A)$ and $(B,B)$.
In that case, both
properties~(a) and (b) of Definition~\ref{D2.2} hold,
while property~(c) is vacuous. In other words, the Stag-Hunt
game is also a coordination game with respect to the new selection
of the desirable set $\oA\,$.

\begin{claim}\label{CL2.1}
In any coordination game and for any action profile
$\alpha\notin\cA^*\cup\oA$ there exists a
sequence of action profiles $\{\alpha^{k}\}$, such that
$\alpha^{0}=\alpha$ and $\alpha^{k}_{i}\in\mathrm{BR}_{i}(\alpha^{k-1})$
for some $i$, terminates at an action profile in $\cA^*\cup\oA\,$.
\end{claim}

\begin{proof}
By Definition~\ref{D2.2}~(b)
there exists an agent $i\in\cI$ and an action
$\alpha_{i}^{1}\in\mathrm{BR}_{i}(\alpha^0)$, such that
$u_{i}\left(\alpha_{i}^1,\alpha_{-i}^0\right) >
u_{i}\left(\alpha_{i}^0,\alpha_{-i}^0\right)$ and
$u_{s}\left(\alpha_{i}^1,\alpha_{-i}^0\right)\ge
u_{s}\left(\alpha_{i}^0,\alpha_{-i}^0\right)$ for all
$s\ne{i}\,\,.$ Define $\alpha^1\df (\alpha^1_{i}, \alpha^0_{-i})$.
Unless $\alpha^1\in \cA^*\cup\oA$, we can repeat the same argument
to generate an action profile $\alpha^2$ and so on. Thus, we
construct a sequence $(\alpha^0,\alpha^1,\alpha^2,\dotsc)$ along
which the map $\alpha\mapsto\sum_{i\in\cI} u_{i}(\alpha)$ is
strictly monotone. However, since $\cA$ is finite, the sequence must
necessarily terminate at some $\alpha^k\in\cA^*\cup\oA$ for
$k<\abs{\cA}$.
\end{proof}

Note that when $\oA\subseteq\cA^*$, then
a direct consequence of Claim~\ref{CL2.1}
is that coordination games are weakly acyclic games
(cf.,~\cite{Young04}).

\subsection{Network Formation Games}\label{Sec:NFG}

Network formation games are of particular interest in wireless
communications due to their utility in modeling distributed topology
control \cite{Santi05} and overlay routing \cite{Chun04}. Recent
developments in distributed learning dynamics, e.g.,
\cite{ChasparisShamma_CDC08}, have also provided the tools for
computing efficient solutions for these games in a distributed
manner.

To illustrate how a network formation game can be modeled as a
coordination game, we introduce a simple network formation game
motivated by \cite{JacksonWolinsky96}.
Let us consider $n$ nodes deployed on the plane and assume that the
set of actions of each agent $i$, $\cA_{i}$, contains all
possible combinations of neighbors of $i$, denoted $\mathcal{N}_{i}$,
with which a link can be established, i.e.,
$\cA_{i}=2^{\mathcal{N}_{i}}$. Links are considered
unidirectional, and a link established by node $i$ with node $s$,
denoted $(s,i)$, starts at $s$ with the arrowhead pointing to $i$.
A \emph{graph} $G$ is defined as a collection of nodes and
directed links. Define also a \emph{path} from $s$ to $i$ as a
sequence of nodes and directed links that starts at $s$ and ends to
$i$ following the orientation of the graph, i.e.,
$$(s\rightarrow{i}) =
\bigl\{s=s_{0},(s_{0},s_{1}),s_{1},\dotsc,(s_{m-1},s_{m}),s_{m}=i\bigr\}$$
for some positive integer $m$.
In a \emph{connected} graph, there is a path from any
node to any other node.

Let us consider the utility function
$u_{i}:\cA\rightarrow\mathbb{R}$, $i\in\cI$, defined by
\begin{equation}\label{E-NFG}
 u_{i}(\alpha) \df
 \sum_{s\in\cI\setminus\{i\}}\chi_{\alpha}(s\rightarrow{i})
 - c \magn{\alpha_{i}},
\end{equation}
where $\magn{\alpha_{i}}$ denotes the number of links corresponding
to $\alpha_{i}$ and $c$ is a constant in $(0,1)$. Also,
\begin{equation*}
 \chi_{\alpha}(s\to{i}) \df \begin{cases}
 1 & \mbox{if } (s\to{i})\subseteq G_{\alpha}\,,\\
 0 & \mbox{otherwise,}
 \end{cases}
\end{equation*}
where $G_\alpha$ denotes the graph induced by joint action $\alpha$.
The resulting Nash equilibria are usually called \emph{Nash networks}
\cite{Bala00}.
As it was shown in Proposition~4.2 in
\cite{ChasparisShamma_CDC08}, a network $G^*$ is a Nash network if
and only if it is \emph{critically connected}, i.e., i) it is
connected, and ii) for any $(s,i)\in{G}$, $(s\rightarrow i)$ is the unique path
from $s$ to $i$. For example, the resulting Nash networks for $n=3$
agents and unconstrained neighborhoods are shown in
\figurename~\ref{Fig:NN}.

\begin{figure}[!ht]
\centering
\includegraphics[width=2.2in]{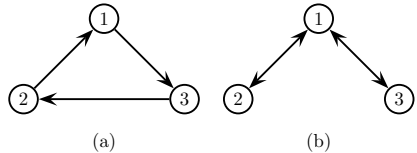}
\caption{Nash networks in case of $n=3$ agents and $0<c<1$.}
\label{Fig:NN}
\end{figure}

Let us define $\oA$ to be the following set of
action profiles
\begin{equation*}
 \oA \df \{\alpha^{*}\in\cA:
 u_{i}(\alpha^*) = \max_{\alpha\in\cA}\;u_{i}(\alpha)
 \text{~for all~} i\in\cI\}\,,
\end{equation*}
which corresponds to the set of \emph{payoff-dominant} networks.
Note that payoff-dominant networks (if they exist) are connected
with minimum number of links.
Also, \emph{not} all Nash networks are necessarily payoff-dominant.
For example, in \figurename~\ref{Fig:NN}(a),
assuming that $0<c<1$, all players realize the same utility, which
is equal to $2 - c$. This is a strict Nash network since each agent
can only be worse off by unilaterally changing its links. It is also
the payoff-dominant network. On the other hand, \figurename~\ref{Fig:NN}(b)
is a non-strict Nash network and is payoff-dominated by
\figurename~\ref{Fig:NN}(a).

The utility function \eqref{E-NFG}
corresponds to the \emph{connections model} of
\cite{JacksonWolinsky96} and has been used to describe various
economic and social contexts such as transmission of information. It
has also been applied for distributed topology control in wireless
networks \cite{Komali08}. Practically, it constitutes a measure of
network connectivity, since the maximum utility for node $i$ is
achieved when there is a path from any other node to $i$.

\begin{claim}\label{CL2.2}
The network formation game defined by
\eqref{E-NFG} is a coordination
game, provided the set of payoff-dominant networks is non-empty.
\end{claim}

\begin{proof}
For a joint action $\alpha\notin\cA^*$ suppose that an agent $i$
picks the \emph{best reply} in
$\mathrm{BR}_{i}(\alpha)\neq\varnothing$ (i.e., the most profitable
better reply). Then no other agent becomes worse off, since a best
reply for $i$ always retains connectivity. Note that this is not
necessarily true for any other better reply. Thus,
Definition~\ref{D2.2}~(b) is satisfied. In order to show
property~(c), consider any joint action $\alpha$ that is a Nash
network. If any one agent $j_{1}$ selects the action
$\Tilde{\alpha}_{j_{1}}$ of establishing ``no links", then there
exists at least one other agent $j_{2}\neq{j_{1}}$ whose payoff
becomes strictly less than the equilibrium payoff (e.g., pick
$j_{2}$ such that $(j_{1},j_{2})\in G_\alpha$). This is due to the
fact that $\alpha$ is critically connected. Continue in the same
manner by selecting $\Tilde{\alpha}_{j_{2}}$ to be the action of
establishing ``no links", and so on. This way, we may construct a
sequence of agents and an action profile which satisfies
Definition~\ref{D2.2}~(c) of a coordination game.
\end{proof}

The condition that payoff-dominant networks exist is not
restrictive. For example, if
$\mathcal{N}_{i}=\cI\setminus\{i\}$ for all $i$, then the set of
\emph{wheel networks} (cf.,~\cite{ChasparisShamma_CDC08}) is payoff dominant.

In a forthcoming section, we present a distributed optimization
approach for achieving convergence to payoff-dominant networks
through aspiration learning which is of independent interest.

\subsection{Common-Pool Games}\label{Sec:CPG}

Common-pool games refer to strategic interactions where two or more
agents need to decide unilaterally whether or not to utilize a
limited common resource. In such interactions, each agent would
rather use the common resource by itself than share it with another
agent, which is usually penalizing for both.

We define common-pool games as follows:

\begin{definition}[Common-Pool Game]
A common-pool game is a strategic-form game such that for each agent
$i\in\cI$, $\cA_{i}=\{p_{0},p_{1},\dotsc,p_{m-1}\}$,
with $0\le p_{0}<p_{1}<\dotsb<p_{m-1}$, and
\begin{equation*}
u_{i}(\alpha) \df \begin{cases}
1-c_{j}\,, & \text{if~}\alpha_{i}=p_{j} \text{~and~} \alpha_{i}
> \max_{\ell\neq{i}}\alpha_{\ell}\,,\\[3pt]
-c_{j} + \tau_{j}\,, & \text{if~}\alpha_{i}=p_{j} \text{~and~}
\exists s\in\cI\setminus\{i\}
\mbox{~s.t.~} \alpha_{s} > \max_{\ell\neq{s}}\alpha_{\ell}\,,\\[3pt]
-c_{j}\,, & \text{if~}\alpha_{i}=p_{j} \text{~and~}
\nexists s\in\cI \mbox{~s.t.~} \alpha_{s} > \max_{\ell\neq{s}}\alpha_{\ell}\,,
 \end{cases}
\end{equation*}
where $0\le c_{0}< \dotsb < c_{m-1} < 1$, $\tau_{j} > 0$ for all
$j=0,1,\dotsc,m-2$, and
$$-c_{0} < -c_{m-1} + \tau_{m-1} < \dotsb < -c_{0}+\tau_{0} < 1-c_{m-1}\,.$$
\end{definition}

This definition of a common-pool game can be viewed as a
finite-action analog of continuous-action common-pool games defined
in \cite{Meinhardt99}.
Table~\ref{Tb:CPG} presents an example of a common-pool game of 2 players
and 3 actions.
\begin{table}[!ht]
\centering\footnotesize
\begin{game}{3}{3}
 & $p_{0}$ & $p_{1}$ & $p_{2}\strut$\\
$p_{0}$ & $-c_{0},-c_{0}$ & $-c_{0}+\tau_{0},1-c_{1}$ &
$-c_{0}+\tau_{0}, 1-c_{2}\strut$ \\
$p_{1}$ & $1-c_{1},-c_{0}+\tau_{0}$ & $-c_{1},-c_{1}$ &
$-c_{1}+\tau_{1}, 1-c_{2}\strut$ \\
$p_{2}$ & $1-c_{2},-c_{0}+\tau_{0}$ & $1-c_{2},-c_{1}+\tau_{1}$ &
$-c_{2},-c_{2}\strut$ \\
\end{game}\\[5pt]
\caption{A common-pool game of 2 players and 3 actions.}\label{Tb:CPG}
\end{table}

We call ``successful'' any action profile in which one player's
action is strictly greater than any other player's action. Any other
situation corresponds to a ``failure.'' In common-pool games, we
define the set of desirable action profiles $\oA$, as the set of
successful action profiles, i.e.,
\begin{equation}\label{E-CPG}
\oA \df \Bigl\{\alpha\in\cA: \exists i\in\cI \mbox{ s.t. } \alpha_{i} >
\max_{\ell\neq{i}}\;\alpha_{\ell}\Bigr\}\,.
\end{equation}

For example, this set of joint actions corresponds to the
off-diagonal action profiles in Table~\ref{Tb:CPG}.
Moreover, the set $\oA$ payoff-dominates the set $\cA\setminus\oA\,$.

\begin{claim}
Any common-pool game is a strict coordination game.
\end{claim}

\begin{proof}
Let $\oA$ be defined as in \eqref{E-CPG}.
Note first that for any $\alpha^*\in\oA$ and $\alpha\in\cA\setminus\oA\,$, we have
$u_{i}(\alpha^*) > u_{i}(\alpha)$ for all $i\in\cI$.
In other words, Definition~\ref{D2.2}~(a) is satisfied.

Moreover, note that any $\alpha\notin\oA$ is not a Nash equilibrium.
For any action profile $\alpha\notin\oA\,$, pick an agent $i$ such that
$i\in\arg\max_{s\in\cI}\alpha_{s}$.
Let us also assume that $\alpha_{i}=p_{j}$ for some $j\in\{0,1,\dotsc,m-1\}$.
If $j>0$, then agent $i$ can increase its utility by selecting action $p_{k}$
for any $k<j$.
In that case, the utility of any other agent either increases or remains the same.
If, instead, $j=0$, then agent $i$ can increase its utility by selecting
action $p_{k}$ for any $k>j$.
In this case, the utility of any other agent increases.
Thus, Definition~\ref{D2.2}~(b) is also satisfied.

Lastly, note that $\cA^*\subseteq\oA$.
To check this, consider any $\alpha\notin\oA$.
As the previous discussion revealed, there always exist an agent and a
better reply for that agent, i.e., $\cA^*\subseteq\oA$.
Thus, Definition~\ref{D2.2}~(c) is trivially satisfied.
\end{proof}

If we imagine that a common-pool game is played repeatedly over
time, it would be desirable that i) failures are avoided, and ii)
agents manage to equally share the time they succeed (i.e., access
the common resource). In other words, convergence to a successful
state may not be sufficient. Instead, a (possibly time-dependent)
solution that equally divides the time-slots that each user utilizes
the common resource would seem more appropriate.

Distributed convergence to such solutions is currently an open issue
in packet radio multiple-access protocols (see, e.g.,
\cite[Chapter~5]{Han08}). In these scenarios, there are multiple
users that compete for access to a single communication channel.
Each user needs to decide whether or not to occupy the channel in a
given time-slot based only on \emph{local} information. If more than
one user is occupying the channel, then a \emph{collision} occurs
and the user needs to resubmit the data. An example of such
multiple-access protocol is the Aloha protocol \cite{Abramson70},
where users decide on transmitting a packet according to a
probabilistic pattern. In this line of work, the action space of
each user consists of multiple power levels of transmission
\cite{Tembine09}. If a user transmits with a power level that is
strictly larger than the power level of any other user, then it is
able to transmit successfully, otherwise a collision occurs and
transmission is not possible. This game can be formulated in a
straightforward manner as a common-pool game.

In a forthcoming section we provide a distributed solution to this
problem using aspiration learning which is of independent interest.

\section{Aspiration Learning}\label{S3}

In this section, we define aspiration learning, motivated by
\cite{Karandikar98}. For some constants $\zeta>0$, $\epsilon>0$,
$\lambda\geq{0}$, $c>0$, $0<h<1$, and
$\underline{\rho},\overline{\rho}\in\mathbb{R}$, such that
$$-\infty < \underline{\rho} <
\min_{\alpha\in\cA,\,i\in\cI}\;u_{i}(\alpha)\le
\max_{\alpha\in\cA,\,i\in\cI}\;u_{i}(\alpha) <
\overline{\rho} < \infty\,,$$ the aspiration learning iteration
initialized at $(\alpha(0),\rho(0))$ is described in
Table~\ref{Tb:AL}.

\begin{table}[!ht]
\fbox{ {\small
\begin{minipage}{0.96\textwidth}
At every $t=0,1,\dotsc$, and for each $i\in\cI$

\begin{enumerate}
\item
Agent $i$ plays $\alpha_{i}(t)$ and measures utility
 $u_{i}(\alpha(t))$.

\item
Agent $i$ updates its aspiration level according to
\begin{equation*}
\rho_{i}(t+1) = \mathtt{sat}
\bigl[\rho_{i}(t) + \epsilon[u_{i}(\alpha(t)) - \rho_{i}(t)]
+ r_{i}(t)\bigr]\,,
\end{equation*}
where
\begin{equation*}
r_{i}(t) \df \begin{cases}
0\,, & \mbox{ w.p. } 1-\lambda\,, \\[2pt]
\mathtt{rand}[-\zeta,\zeta]\,, & \mbox{ w.p. } \lambda\,,
\end{cases}
\end{equation*}
and
\begin{equation*}
\mathtt{sat}[\rho] \df \begin{cases}
\overline{\rho}\,, & \text{if~}\rho>\overline{\rho}\,,\\
\rho\,, & \text{if~}\rho\in[\underline{\rho},\overline{\rho}]\,,\\
\underline{\rho}\,, & \text{if~}\rho < \underline{\rho}\,.
\end{cases}
\end{equation*}

\item
Agent $i$ updates its action:
\begin{equation*}
\alpha_{i}(t+1) = \begin{cases}
\alpha_{i}(t) & \mbox{w.p.~} \phi\bigl(u_{i}(\alpha(t))-\rho_{i}(t)\bigr)\,,
\\[2pt]
\mathtt{rand}(\cA_{i}\setminus\alpha_{i}(t)) & \mbox{w.p.~}
1-\phi\bigl(u_{i}(\alpha(t))-\rho_{i}(t)\bigr)\,,
\end{cases}
\end{equation*}
where
\begin{equation*}
\phi(z) \df \begin{cases}
1\,, & \text{if~}z\ge0\,,\\[2pt]
\max(h,1+cz)\,, & \text{if~}z<0\,.\end{cases}
\end{equation*}

\item
Agent $i$ updates the time and repeats.
\end{enumerate}
\end{minipage}
} }\\[5pt]
\caption{Aspiration Learning}\label{Tb:AL}
\end{table}

According to this algorithm, each agent $i$ keeps track of an
\emph{aspiration level} $\rho_{i}$, which measures player $i$'s
desirable return and is defined as a perturbed fading memory
average of its payoffs throughout the history of play.

Given the current aspiration level $\rho_{i}(t)$, agent $i$ selects
a new action $\alpha_{i}(t+1)$. If the previous action $\alpha_{i}(t)$
provided utility at least $\rho_{i}(t)$, then the agent is
``satisfied" and repeats the same action, i.e.,
$\alpha_{i}(t+1)=\alpha_{i}(t)$. Otherwise, $\alpha_{i}(t+1)$ is
selected randomly over all available actions, where the probability
of selecting again $\alpha_{i}(t)$ depends on the level of discontent
measured by the difference $u_{i}(\alpha(t))-\rho_{i}(t)<0$. The random
variables $\{r_{i}(t): t\ge0\,,\;i\in\cI\}$ are independent,
identically distributed and are referred to as the ``tremble.''

Let $\cX\df \cA \times[\underline{\rho},\overline{\rho}]^{n}$, i.e.,
pairs of joint
actions $\alpha$ and vectors of aspiration levels, $\rho_{i}$, $i\in\cI$.
The set $\cA$ is endowed with the product topology,
$[\underline{\rho},\overline{\rho}]$ with its usual Euclidean topology,
and $\cX$ with the corresponding product topology.
We also let $\Bor(\cX)$ denote the Borel $\sigma$-field of $\cX$, and
$\Pm(\cX)$ the set of probability measures on $\Bor(\cX)$ endowed with the Prohorov topology, i.e.,
the topology of weak convergence.
The algorithm in Table~\ref{Tb:AL} defines
an $\cX$-valued Markov chain.
Let $P_{\lambda}:\cX\times\Bor(\cX)\to[0,1]$ denote its transition
probability function, parameterized by $\lambda>0$.
We refer to the process with $\lambda>0$
as the \emph{perturbed process}.

We let $\Cc(\cX)$ denote the Banach space of real-valued
continuous functions on $\cX$ under the sup-norm (denoted
by $\supnorm{\,\cdot\,}$) topology.
For $f\in\Cc(\cX)$ we define
$$P_{\lambda} f(x) \df \int_{\cX}P_{\lambda} (x,\D{y})f(y)\qquad\text{and}
\qquad
\mu[f]\df\int_{\cX}\mu(\D{x})f(x)\,,\quad \mu\in\Pm(\cX)\,.$$
It is straightforward to verify that $P_{\lambda}$ has the Feller property,
i.e., $P_{\lambda} f\in\Cc(\cX)$ for all $f\in\Cc(\cX)$.
Recall that $\mu_{\lambda}\in\Pm(\cX)$ is called an invariant probability
measure for $P_{\lambda}$ if
$$ (\mu_{\lambda}P_{\lambda})(A)\df\int_{\cX}
\mu_{\lambda}(\D{x}) P_{\lambda}(x,A)
= \mu_{\lambda}(A)\qquad \forall A\in\Bor(X)\,.$$
Since $\cX$ is a compact metric space
and $P_{\lambda}$ has the Feller property it
admits an invariant probability measure $\mu_{\lambda}$
\cite[Theorem~7.2.3]{Lerma03}.

We are interested in the asymptotic behavior of the aspiration learning
algorithm
as the ``experimentation probability'' $\lambda$ approaches zero.
We say that a state $x\in\cX$ is \emph{stochastically stable} if any
collection of invariant probability measures
$\{\mu_{\lambda}\in\Pm(\cX) : \mu_{\lambda}P_{\lambda}
=\mu_{\lambda}\,,\; \lambda>0\}$ satisfies
$\liminf_{\lambda\downarrow{0}}\;\mu_{\lambda}(x)>0$. It
turns out that the stochastically stable states comprise a finite
subset of $\cX$ which is defined next.

\begin{definition}
A \emph{pure strategy state} is a state
$s=(\alpha,\rho)\in\cX$ such that for all $i\in\cI$,
$u_{i}(\alpha)=\rho_{i}$.
The set of pure strategy states is denoted by $\cS$
and $\magn{\cS}$ denotes its cardinality.
\end{definition}

Note that the set $\cS$ is isomorphic to $\cA$ and can be identified
as such.

As customary, the Dirac measure in $\Pm(\cX)$ supported at $x\in\cX$ is denoted by
$\delta_{x}$.
The objective in this section is to characterize the set of
stochastically stable states.
Our main result is summarized in the following theorem:

\begin{theorem}\label{T3.1}
There exists a unique probability vector
$\pi=(\pi_{1},\dotsc,\pi_{\abs{\cS}})$ such that for any
collection of invariant probability measures
$\{\mu_{\lambda}\in\Pm(\cX) : \mu_{\lambda}P_{\lambda}
=\mu_{\lambda}\,,\;\lambda>0\}$, we have
\begin{equation*}
\lim_{\lambda\downarrow0}\; \mu_{\lambda}(\cdot)
= \Hat{\mu}(\cdot) \df \sum_{s\in\cS}{\pi}_{s}\delta_{s}(\cdot)\,,
\end{equation*}
where convergence is in the weak$^{*}$ sense.
\end{theorem}

As we show later, $\pi$ in Theorem~\ref{T3.1} is the unique
invariant distribution of a finite-state Markov chain.

\begin{remark}	\label{Rm:Ergodicity}
The expected asymptotic behavior of aspiration learning can be
characterized by $\Hat{\mu}$ and, therefore, $\pi$. In particular,
by Birkhoff's individual ergodic theorem, e.g.,
\cite[Theorem~2.3.4]{Lerma03}, and the weak convergence of
$\mu_{\lambda}$ to $\Hat{\mu}$, the expected percentage of time that
the process spends in any $B\in\mathcal{B}(\cX)$ such that
$\partial{B}\cap\oS\ne\varnothing$ is given by $\Hat{\mu}(B)$ as the
experimentation probability $\lambda$ approaches zero and time
increases, i.e.,
\begin{equation*}
\lim_{\lambda\downarrow{0}}\left(\lim_{t\to\infty}\;
\frac{1}{t}\sum_{k=0}^{t-1}P_{\lambda}^{k}(x,B)\right) = \Hat{\mu}(B)\,.
\end{equation*}
\end{remark}

The proof of Theorem~\ref{T3.1} requires a
series of propositions, which comprise the remaining of this section.

Let $P(\cdot\,,\cdot)$ denote the transition probability function
on $\cX\times\Bor(\cX)$ corresponding to $\lambda=0$.
We refer to the process $\{X_{t} : t\ge0\}$ governed by $P$ as the
\emph{unperturbed process}.
Let $\Omega\df\cX^{\infty}$ denote the canonical path space, i.e., an
element $\omega\in\Omega$ is a sequence $\{\omega(0),\omega(1),\dotsc\}$, with
$\omega(t)= (\alpha(t),\rho(t))\in\cX$.
We use the same notation for the elements $(\alpha,\rho)$ of the space
$\cX$ and for the coordinates of the process $X_{t}=(\alpha(t),\rho(t))$.
Let also $\Prob_{x}$ denote the unique
probability measure induced by $P$ on the
product $\sigma$-algebra of $\cX^{\infty}$,
initialized at $x=(\alpha,\rho)$,
and $\Exp_{x}$ the corresponding expectation operator.
Let also $\sF_{t}\df \sigma(X_{\tau}\,,~ \tau\le t)\,,$ $t\geq{0}$, denote 
the $\sigma$-algebra generated by $\{X_{\tau},~\tau\le{t}\}$.

For $t\ge0$ define the sets
\begin{align*}
A_{t} &\df
\left\{\omega\in\Omega:\alpha(\tau)=\alpha(t)\,, \text{~for all~} \tau
\geq t \right\}\,,\\[5pt]
B_{t} &\df \{\omega\in\Omega:\alpha(\tau)=\alpha(0)\,, \text{~for all~}
0\le\tau\le{t}\}\,.
\end{align*}
Note that $\{B_{t}:t\ge0\}$ is a non-increasing sequence, i.e.,
$B_{t+1}\subseteq B_{t}$, while $\{A_{t}:t\ge0\}$ is non-decreasing.
Recall that the shift operator $\theta_{t}:\Omega\to\Omega$, $t\ge0$,
satisfies $X_{s}(\theta_{t}(\omega))=X_{s+t}(\omega)$.
Therefore $A_{t} = \theta^{-1}_{t}(B_{\infty})$.
Let $A_{\infty} \df\bigcup_{t=0}^{\infty}A_{t}$ and
$B_{\infty}\df\bigcap_{t=1}^{\infty}B_{t}$.
The set $A_{\infty}$ is the event that agents
eventually play the same action profile, while $B_{\infty}$
is the event that agents never change their actions.
For $D\in\Bor(\cX)$ we let
$\uptau(D)$ denote the first hitting time of $D$, i.e.,
\begin{equation}\label{E-uptau}
\uptau(D) \df \inf\; \{t\ge 0 : X_{t}\in D\}\,.
\end{equation}

\begin{proposition}\label{P3.1}
It holds that
$$\inf_{x\in\cX}\;\mathbb{P}_{x}(B_{\infty})>0\qquad\text{and}\qquad
\inf_{x\in\cX}\;\Prob_{x}(A_\infty)=1\,.$$
\end{proposition}
\begin{proof}
Assume that the process is initialized at $X_{0}=x=(\alpha,\rho)$.
Note that $B_{t}$ consists of those sample paths which satisfy
\begin{equation*}
\rho_{i}(\tau) =
u_{i}(\alpha) - (1-\epsilon)^{\tau} \bigl(u_{i}(\alpha)-\rho\bigr)\,,
\quad 0\le\tau<t\,,\quad i\in\cI\,.
\end{equation*}
Therefore, we have:
\begin{equation}\label{E-P3.1a}
\mathbb{P}_{x}(B_{t}) = \prod_{0\le\tau<t}\;\prod_{i\in\cI}
\max\left\{h,1-c(1-\epsilon)^\tau\bigl(\rho_{i}-u_{i}(\alpha)\bigr)^{+}\right\},
\end{equation}
where
\begin{eqnarray*}
(x)^{+} \df \begin{cases} x\,, & \text{if~} x\ge0\,,\\
0\,, & \text{othewise.}
\end{cases}
\end{eqnarray*}
Let $T_{0}$ satisfy
$c(1-\epsilon)^{T_{0}}(\overline{\rho}-\underline{\rho})\le
\min\;\{1- h, \epsilon\}\,.$
Then
\begin{align*}
\mathbb{P}_{x}(B_{t}) &\ge h^{nT_{0}} \prod_{i\in\cI}\;\prod_{T_{0}<\tau<t}
\Bigl(1-c(1-\epsilon)^\tau
\bigl(\rho_{i}-u_{i}(\alpha)\bigr)^{+}\Bigr)\notag\\[5pt]
&\ge h^{nT_{0}} \prod_{i\in\cI}\;\Biggl(1- c\bigl(\rho_{i}-u_{i}(\alpha)\bigr)^{+}
\sum_{\tau=T_{0}+1}^{t} (1-\epsilon)^\tau\Biggr)\notag\\[5pt]
&\ge h^{nT_{0}}
\prod_{i\in\cI}\;\left( 1- (1-\epsilon)\,
\frac{\bigl(\rho_{i}-u_{i}(\alpha)\bigr)^{+}}
{\overline{\rho}-\underline{\rho}}\right)\qquad \forall t>T_{0}\,,
\end{align*}
and since the sequence $\{B_{t}\}$ is non-increasing, also for all $t\ge0$.
Therefore, by continuity from above, we obtain
$\inf_{x\in\cX}\; \Prob_{x} (B_{\infty}) \ge
\epsilon^{n}h^{nT_{0}}\,,$
which proves the first claim.

Next, define the set
\begin{equation*}
D_{\ell} \df \left\{(\alpha,\rho)\in\cX :
\rho_{i} - u_{i}(\alpha)
\le (1-\epsilon)^{\ell}\bigl(\overline{\rho}-\underline{\rho}\bigr)\,,
~\forall i\in\cI\right\}\,,\qquad \ell\ge0\,,
\end{equation*}
and note that
$\Prob_{x}(B_{\ell}) \le P^{\ell}(x,D_{\ell})$,
where $P^{t}$, $t\ge0$, denotes the multistage transition probability function
defined by the recursion $P^{t}=P^{t-1}P$ and $P^{0}=I$.
Thus, using the Markov property over $k$ time blocks of length $\ell$, we obtain
the rough estimate
\begin{align}\label{E-P3.1b}
\Prob_{x}(\uptau(D_{\ell})>k\ell) &\le
\Prob_{x}(X_{j\ell}\in D^{c}_{\ell}\,,~j=1,\dotsc,k)\nonumber\\
&\le \Prob_{x}(X_{j\ell}\in D^{c}_{\ell}\,,~j=1,\dotsc,k-1)\,
\biggl(\sup_{z\in D^{c}_{\ell}}\;P^{\ell}(z,D^{c}_{\ell})\biggr)\nonumber\\
&\le \left(1-\inf_{z\in\cX}\;\Prob_{z}(B_{\ell})\right)\,
\Prob_{x}(X_{j\ell}\in D^{c}_{\ell}\,,~j=1,\dotsc,k-1)\,.
\end{align}
Let $q_{0} \df 1-\inf_{z\in\cX}\;\Prob_{z}(B_{\infty})$.
We have already shown that $q_{0}<1$.
Finite induction on \eqref{E-P3.1b} yields
\begin{equation*}
\Prob_{x}(\uptau(D_{\ell})>k\ell) \le
\Bigl(1 - \inf_{z\in\cX}\;\Prob_{z}(B_{\ell})\Bigr)^k\le q_{0}^{k}\,.
\end{equation*}
We have
\begin{equation*}
\Prob_{x}(A_{k\ell}) \ge \sum_{t=1}^{k\ell}
\Prob_{x}\bigl(\uptau(D_{\ell})=t,\; X\circ\theta_{t}\in B_\infty\bigr)
\end{equation*}
and thus using the Markov property together with the fact that
$X_{\uptau(D_{\ell})}\in D_{\ell}$ a.s.\ on $\{\uptau(D_{\ell})<\infty\}$,
and setting $k=\ell$,
we obtain
\begin{align}\label{E-P3.1c}
\Prob_{x}(A_{\ell^{2}}) &\ge \sum_{t=1}^{\ell^{2}}
\Prob_{x}\bigl(\uptau(D_{\ell})=t\bigr)
\inf_{y\in D_{\ell}}\;\Prob_{y}(B_\infty)\notag\\[5pt]
&\ge\Bigl(1-\Prob_{x}(\uptau\bigl(D_{\ell})>\ell^{2}\bigr)\Bigr)
\inf_{y\in D_{\ell}}\;\Prob_{y}(B_\infty)\notag\\[5pt]
&\ge\left(1-q_{0}^{\ell}\right)
\inf_{y\in D_{\ell}}\;\Prob_{y}(B_\infty)\,.
\end{align}
It is clear by \eqref{E-P3.1a} that
$\inf_{x\in D_{\ell}}\;\Prob_{x}(B_{\infty})\to1$ as $\ell\to\infty$.
Therefore both terms on the right hand side of \eqref{E-P3.1c} converge
to $1$ as $\ell\to\infty$, and the proof is complete.
\end{proof}

\begin{proposition}\label{P3.2}
There exists a transition probability function $\Pi$ on $\cX\times\Pm(\cX)$
that has the Feller property and $\Pi(x,\cdot)$ is supported on $\cS$
for all $x\in\cX$, and such that
\begin{itemize}
\item[\upshape{(i)}]
For all $f\in\Cc(\cX)$,
$\lim_{t\rightarrow\infty}\;\supnorm{P^{t}f-\Pi{f}}=0$.
\item[\upshape{(ii)}]
If $R_{\lambda}$ is a resolvent of $P$, defined by
$$R_{\lambda} \df
\varphi(\lambda)\sum_{t=0}^{\infty}(1-\varphi(\lambda))^{t}P^{t}\,,$$
where $\varphi(\lambda)\in(0,1)$, $\lambda>0$, and
$\lim_{\lambda\rightarrow{0}}\varphi(\lambda)=0$, then
$$\lim_{\lambda\rightarrow{0}}\;\supnorm{R_{\lambda}{f}-\Pi{f}}=0\qquad
\forall f\in\Cc(\cX)\,.$$
\end{itemize}
\end{proposition}

\begin{proof}
For $f\in\Cc(\cX)$ and $x\in\cX$, we have
$\Exp_{x}[f(X_{t})]= P^{t} f(x)$.
Since $A_{t} = \theta^{-1}_{t}(B_{\infty})$,
then using the Markov property we obtain
that, for any positive $t$ and $t'$,
\begin{align}\label{E-P3.2a}
\babs{P^{2t}f(x)-P^{2t+t'} f(x)}
&= \babs{\Exp_{x}\bigl[f(X_{2t})-f(X_{2t+t'})\bigr]}\notag\\[5pt]
&= \babs{\Exp_{x}\bigl[\bigl(f(X_{2t})-f(X_{2t+t'})\bigr)
\mathbf{1}_{A_{t}}\bigr]}
+ \babs{\Exp_{x}\bigl[\bigl(f(X_{2t})-f(X_{2t+t'})\bigr)
\mathbf{1}_{A^{c}_{t}}\bigr]}\notag\\[5pt]
&\le\Bigl\lvert
\Exp_{x}\Bigl[\Exp\bigl[\bigl(f(X_{2t})-f(X_{2t+t'})\bigr)
\mathbf{1}_{A_{t}}\bigm| \mathfrak{F}_{t}\bigr]\Bigr]\Bigr\rvert
+2 \mathbb{P}_{x}(A^{c}_{t}) \supnorm{f}\notag\\[5pt]
&\le \Exp_{x}\Bigl[\Exp_{X_{t}}
\bigl[\abs{f(X_{2t})-f(X_{2t+t'})}\,\mathbf{1}_{A_{t}}\bigr]\Bigr]
+2 \mathbb{P}_{x}(A^{c}_{t})\supnorm{f}\notag\\[5pt]
&\le \sup_{z\in\cX}\;\Exp_{z}\bigl[\abs{f(X_{t})-f(X_{t+t'})}\,
\mathbf{1}_{B_{\infty}}\bigr]
+2 \mathbb{P}_{x}(A^{c}_{t})\supnorm{f}\,.
\end{align}
Since for any initial condition
$x=(\alpha,\rho)$ the dynamics on $B_{\infty}$ evolve according to
\begin{equation*}
\rho{(t)} = \varrho(t;\alpha,\rho)
\df u(\alpha) - (1-\epsilon)^{t} \bigl(u(\alpha)-\rho\bigr)\,,
\end{equation*}
the continuity of $f$
(which is necessarily uniform since $\cX$ is compact) yields
\begin{multline}\label{E-P3.2b}
\sup_{t'\ge0}\;\sup_{(\alpha,\rho)\in\cX}\;
\Exp_{(\alpha,\rho)}\bigl[\abs{f(X_{t})-f(X_{t+t'})}\,
\Ind_{B_{\infty}}\bigr]\\
=\sup_{t'\ge0}\;\sup_{(\alpha,\rho)\in\cX}\;
\babs{f\bigl(\alpha,\varrho(t;\alpha,\rho)\bigr)
-f\bigl(\alpha,\varrho(t+t';\alpha,\rho)\bigr)}
\xrightarrow[t\to\infty]{}0\,.
\end{multline}
By \eqref{E-P3.2a}--\eqref{E-P3.2b} and
Proposition~\ref{P3.1} we obtain
\begin{equation*}
\sup_{t'>0}\; \supnorm{P^{2t}f-P^{2t+t'}f}
\xrightarrow[t\to\infty]{}0\,.
\end{equation*}
Therefore, the sequence
$\{P^{t}f\,,\;t\in\NN\}$ is Cauchy in
$\bigl(\Cc(\cX),\supnorm{\,\cdot\,}\bigr)$, and hence converges in $\Cc(\cX)$.
Let $\varphi(f)(x)\df \lim_{t\to\infty}\; P^{t}f(x)$.
Then for each $x$, $f\mapsto \varphi(f)(x)$ defines a bounded linear
functional on $\Cc(\cX)$.
It is a positive functional since $\varphi(f)(x)\ge 0$, for $f\ge0$,
and if $\bm{1}$ denotes the constant function equal to $1$,
$\varphi(\bm{1})(x)=1$.
Then, by the Riesz representation theorem,
$\varphi(f)(x)$ is a Borel probability measure on $\cX$ for each $x$.
Denote this by $\Pi(x,\cdot)$.
Since $\varphi:\Cc(\cX)\to\Cc(\cX)$,
it follows that $\Pi$ has the Feller property.
Also, by the definition of $\Pi$, we have
\begin{equation}\label{E-P3.2c}
\supnorm{P^{t}f - \Pi{f}}\xrightarrow[t\to\infty]{} 0 \qquad \forall
f\in\Cc(\cX)\,.
\end{equation}
This proves (i).

Next using a triangle inequality, we have for each $T>0$,
\begin{equation*}
\supnorm{R_{\lambda} f - \Pi{f}}
\le \varphi(\lambda)\sum_{t=0}^{T-1}(1-\varphi(\lambda))^{t}
\supnorm{P^{t} f- \Pi{f}}
+ (1-\varphi(\lambda))^{T}\sup_{t\ge T}\;\supnorm{P^{t} f- \Pi{f}}\,.
\end{equation*}
Letting $\lambda\downarrow0$, we obtain
\begin{equation*}
\supnorm{R_{\lambda} f - \Pi{f}}
\le \sup_{t\ge T}\;\supnorm{P^{t} f- \Pi{f}}\qquad \forall T>0\,,
\end{equation*}
and (ii) follows by \eqref{E-P3.2c}.
\end{proof}

We can decompose the transition probability function of the perturbed
process as
\begin{equation}\label{E-Plambda}
P_{\lambda}=(1-\varphi(\lambda))P+\varphi(\lambda)Q_{\lambda}\,,
\qquad \varphi(\lambda)\df 1-(1-\lambda)^{n}\,,
\end{equation}
where $\varphi(\lambda)$ is the probability that at least one agent trembles,
and satisfies $\varphi(\lambda)\downarrow{0}$ as $\lambda\downarrow{0}$.
Also, define the ``lifted" transition probability function:
\begin{equation*}
P_{\lambda}^{L} \df \varphi(\lambda)\sum_{t=0}^{\infty}(1-\varphi(\lambda))^{t}
Q_{\lambda}P^{t} = Q_{\lambda} R_{\lambda}\,,
\end{equation*}
where $R_{\lambda}$ was defined in Proposition~\ref{P3.2}
(the equality on the right-hand side is evident by Fubini).
Similarly we decompose $Q_{\lambda}$ as
\begin{equation*}
Q_{\lambda}=(1-\psi(\lambda))Q+\psi(\lambda)Q^*\,,\qquad
\psi(\lambda) \df 1-\frac{n\lambda(1-\lambda)^{n-1}}{1-(1-\lambda)^{n}}.
\end{equation*}
Here $Q$ is the transition probability function induced by aspiration
learning where exactly one player trembles, and $Q^*$ is the transition
probability function where at least two players tremble simultaneously.



We have the following proposition.

\begin{proposition}\label{P3.3}
The following hold,
\begin{itemize}
\item[\upshape{(i)}]
For $f\in\Cc(\cX)$,
$\lim_{\lambda\rightarrow{0}}\;\supnorm{P_{\lambda}^{L}f-Q\Pi{f}}=0$.

\item[\upshape{(ii)}]
Any invariant distribution $\mu_{\lambda}$ of $P_{\lambda}$
is also an invariant distribution of $P_{\lambda}^{L}$.

\item[\upshape{(iii)}]
Any weak limit point in $\Pm(\cX)$ of $\mu_{\lambda}$,
as $\lambda\downarrow{0}$, is an invariant probability measure of $Q\Pi$.
\end{itemize}
\end{proposition}

\begin{proof}
(i) We have
\begin{align}\label{E-P3.3a}
\supnorm{P_{\lambda}^{L}f-Q\Pi{f}} &\le
\supnorm{Q_{\lambda}(R_{\lambda} f-\Pi{f})}
+\supnorm{Q_{\lambda}\Pi{f} - Q \Pi{f}}\notag\\[5pt]
&\le\supnorm{R_{\lambda} f-\Pi{f}}
+\supnorm{Q_{\lambda}\Pi{f} - Q \Pi{f}}\,.
\end{align}
The first term on the right hand side of \eqref{E-P3.3a}
tends to $0$ as $\lambda\downarrow0$ by
Proposition~\ref{P3.2}, while the second
term does the same by the definition of $Q_{\lambda}$.

(ii) Multiplying both sides of \eqref{E-Plambda}
by $R_{\lambda}$, we have
\begin{equation} \label{E-P3.3b}
P_{\lambda} R_{\lambda} = R_{\lambda} - \varphi(\lambda)I
+ \varphi(\lambda)P_{\lambda}^{L},
\end{equation}
where $I$ denotes the identity operator. Let $\mu_{\lambda}$ denote an
invariant distribution of $P_{\lambda}$. Hence, by \eqref{E-P3.3b},
we have
$$\mu_{\lambda} R_{\lambda} = \mu_{\lambda} R_{\lambda}
- \varphi(\lambda)\mu_{\lambda} +
\varphi(\lambda) \mu_{\lambda} P_{\lambda}^{L}\,,$$
and the second claim follows.

(iii)
Let $\Hat{\mu}$ be a limit point of $\mu_{\lambda}$ as
$\lambda\downarrow0$.
For any $f\in\Cc(\cX)$, we have
\begin{equation*}
\Hat{\mu}[f] - (\Hat{\mu} Q\Pi)[f] =
\bigl(\Hat{\mu}[f] - \mu_{\lambda}[f]\bigr)
+ \mu_{\lambda}\bigl[P_{\lambda}^{L} f - Q\Pi{f}\bigr]
+ \bigl(\mu_{\lambda}\bigl[Q\Pi{f}\bigr] - \Hat{\mu}\bigl[Q\Pi{f}\bigr]\bigr)\,.
\end{equation*}
The first and the third terms on the right hand side tend to $0$
as $\lambda\downarrow0$ along some sequence,
by the weak convergence $\mu_{\lambda}$ to $\Hat{\mu}$,
while the second term is dominated by
$\supnorm{P_{\lambda}^{L}[f] - Q\Pi[f]}$ that also tends to $0$
by part (i).
\end{proof}

For $s\in\cS$ let $N_{\varepsilon}(s)$ denote the
open $\varepsilon$-neighborhood of $s$ in $\cX$.
For any two pure strategy
states, $s,s'\in\cS$, define
$$\Hat{P}_{ss'} \df \lim_{t\rightarrow\infty}QP^{t}(s,N_{\varepsilon}(s'))$$
for some $\varepsilon>0$ sufficiently small.
By Proposition~\ref{P3.1}, $\Hat{P}_{ss'}$ is independent of the selection of
$\varepsilon$.
Define also the $\abs{\cS}\times\abs{\cS}$
stochastic matrix $\Hat{P}\df[\Hat{P}_{ss'}]$.

\begin{proposition}\label{P3.4}
There exists a unique invariant probability
measure $\Hat{\mu}$ of $Q\Pi$.
It satisfies
\begin{equation}\label{E-P3.4a}
\Hat{\mu}(\cdot) = \sum_{s\in\cS}\pi_{s}\delta_{s}(\cdot)
\end{equation}
for some constants $\pi_{s}\geq{0}$, $s\in\cS$.
Moreover, $\pi=(\pi_{1},\dotsc,\pi_{\abs{\cS}})$
is an invariant distribution of $\Hat{P}$, i.e.,
$\pi = \pi\Hat{P}$.
\end{proposition}

\begin{proof}
By Proposition~\ref{P3.2}, the support of
$\Pi$ is $\cS$, and so is the support of $Q\Pi$.
Thus, for any sufficiently small $\varepsilon>0$,
$Q\Pi(s,s') = Q\Pi(s,N_{\varepsilon}(s'))\,.$
Since $Q\Pi$ is a Feller transition function
it admits an invariant probability measure, say $\Hat{\mu}$.
The support of $\Hat{\mu}$ is also $\cS$, and, therefore,
it has the form of \eqref{E-P3.4a} for some constants $\pi_{s}\geq{0}$,
$s\in\cS$.

Note also that $N_\varepsilon(s')$ is a continuity set of $Q\Pi(s,\cdot)$, i.e.,
$Q\Pi(s,\partial{N_\varepsilon}(s'))=0$.
Therefore, by the Portmanteau theorem,
$$Q\Pi(s,N_\varepsilon(s'))=\lim_{t\rightarrow\infty}
QP^{t}(s,N_\varepsilon(s'))= \Hat{P}_{ss'}\,.$$
If we also define $\pi_{s} \df \Hat{\mu}(N_{\varepsilon}(s))$, then
\begin{equation*}
\pi_{s'} = \Hat{\mu}(N_{\varepsilon}(s')) = \sum_{s\in\cS}\pi_{s}
Q\Pi(s,N_{\varepsilon}(s')) = \sum_{s\in\cS}\pi_{s} \Hat{P}_{ss'}\,,
\end{equation*}
which shows that $\pi$ is an invariant distribution of
$\Hat{P}$, i.e., $\pi = \pi\Hat{P}$.

To establish the uniqueness of the invariant distribution of $Q\Pi$,
recall the definition of $Q$. Since $\cS$ is isomorphic with $\cA$,
we can identify $s\in\cS$ with an element $\alpha\in\cA$. If agent
$i$ trembles, then all actions in $\cA_{i}$ have positive
probability of being selected, i.e.,
$Q(\alpha,(\alpha_{i}',\alpha_{-i}))>0$ for all
$\alpha_i'\in\mathcal{A}_i$ and $i\in\cI$. It follows by
Proposition~\ref{P3.1} that
$Q\Pi(\alpha,(\alpha_{i}',\alpha_{-i}))>0$ for all
$\alpha_{i}'\in\cA_{i}$ and $i\in\cI$. Finite induction then shows
that $(Q\Pi)^{n}(\alpha,\alpha')>0$ for all $\alpha$,
$\alpha'\in\cA$. It follows that if we restrict the domain of $Q\Pi$
to $\cS$, then $Q\Pi$ defines an irreducible stochastic matrix.
Therefore, $Q\Pi$ has a unique invariant distribution.
\end{proof}

Theorem~\ref{T3.1} follows from Propositions~\ref{P3.3} and \ref{P3.4}.
Moreover, Proposition~\ref{P3.4} shows that
the unique invariant probability measure of $Q\Pi$ agrees with
the unique invariant probability distribution of the finite
stochastic matrix $\Hat{P}$.

\begin{remark}
A similar result to Proposition~\ref{P3.3}(i), based on which
Theorem~\ref{T3.1} was shown, has also been derived in
\cite[Theorem~2]{Karandikar98}. 
The result in \cite{Karandikar98} though assumes incorrectly that the process
$Q$ satisfies the strong Feller property.
Note that the  proof of Proposition~\ref{P3.3} does not make use of any
such assumption and provides a corrected analysis for the asymptotic behavior
of the aspiration learning scheme presented in \cite{Karandikar98}.
\end{remark}

In the forthcoming sections, we demonstrate the importance of Theorem~\ref{T3.1}
in characterizing the asymptotic behavior of aspiration learning in
large coordination games.
Note that prior analysis of this type of aspiration learning, 
e.g., in \cite{Karandikar98,ChoMatsui05}, was only restricted to two
player and two action games.

\section{Efficiency in Coordination Games} \label{S4}

In this section, we study the asymptotic behavior of the invariant
distribution $\pi$ of $\Hat{P}$ in strict coordination games when
the step size $\epsilon$ approaches zero. The aim is to characterize
the states in $\cS$ that are stochastically stable with respect to
the parameter $\epsilon$. To this end, first denote $\oS$ as the set
of pure strategy states that correspond to $\oA$. Clearly, $\oS$ is
isomorphic to $\oA$. Also, denote by $\cS^*$ the set of pure
strategy states that correspond to the set of Nash action profiles
$\cA^*$.

We define two constants that are important in the analysis:
\begin{equation*}
\begin{split}
\Delta_{\text{min}}&\df \min_{i\in\cI}\;
\min_{\alpha\in\oA\,,\alpha'\notin\oA}\;
\bigl\{u_{i}(\alpha) - u_{i}(\alpha')\bigr\}\\[5pt]
\Delta_{\text{max}} &\df \max_{i\in\cI}\;
\max_{\alpha\neq\alpha'}\;\abs{u_{i}(\alpha') - u_{i}(\alpha)}\,.
\end{split}
\end{equation*}
For strict coordination games $\Delta_{\text{min}}>0$,
and it is the smallest possible payoff decrease
from the dominant payoff due to any deviation from the set of actions in $\oA$.

To facilitate the analysis
we let $\Tilde{\Prob}_{x}$ and $\Tilde{\Exp}_{x}$ denote the probability and
expectation operator, respectively,
on the path space of a Markov process $X_{t}$ starting at $x\in\cX$
at $t=0$, and governed
by the family of transition probabilities $\{QP^{t} : t\ge0\}$.
In other words $\Tilde{\Prob}_{x}(X_{t}\in A)=QP^{t-1}(x,A)$ for
any $A\in\Bor(\cX)$.

\subsection{Two Technical Lemmas}

Lemma~\ref{L4.1} below introduces two new hypotheses.
The first hypothesis corresponds to the case at which payoff
differences within the same action profile are smaller than payoff
differences between dominant and non-dominant action profiles.
The second hypothesis corresponds to the case where each player
receives a unique payoff within $\oA$.

\begin{lemma}\label{L4.1}
Let $\mathscr{G}$ be a strict coordination game
satisfying either one of the following two hypotheses:
\begin{enumerate}
\item[\upshape{(H1)}]
$\delta^{*}\df\max_{i\ne j}\;\max_{\alpha\in\mathcal{A}}\;
\abs{u_{i}(\alpha)-u_{j}(\alpha)}<\Delta_{\text{min}}$.
\item[\upshape{(H2)}]
$\oA\equiv\{\Bar{\alpha}\in\cA:u_{i}(\Bar{\alpha})=\max_{\alpha\in\cA}\;
u_{i}(\alpha)~~\forall i\in\cI\}\,$.
\end{enumerate}
Then, there exists a constant
$C_{0}=C_{0}(\delta^{*},\Delta_{\text{min}},\Delta_{\text{max}})$
such that if $\zeta< C_{0}$ then
\begin{equation*}
\Hat{P}_{\Bar{s}s} \xrightarrow[\varepsilon\downarrow0]{}0
\qquad\text{~for all~} \Bar{s}\in\oS\,,~s\in\cS\setminus\oS\,.
\end{equation*}
\end{lemma}

\begin{proof}
Suppose (H1) holds. Select
$\zeta<\frac{1}{2}(\Delta_{\text{min}}-\delta^{*})$. Let
$x(0)=\Bar{s}\equiv(\Bar{\alpha},\Bar{\rho})\in\oS$. Without loss of
generality suppose agent $1$ trembles. If $r_{1}(0)<0$ the process
clearly converges to $\Bar{s}$ as $t\to\infty$ with probability $1$.
Therefore, suppose $r_{1}(0)>0$. Note that for $t\ge0$ we have
\begin{align}\label{E-L4.1a}
\abs{\rho_{i}(t+1)-\rho_{j}(t+1)} &\le
(1-\epsilon)\abs{\rho_{i}(t)-\rho_{j}(t)}
+\epsilon\abs{u_{i}(\alpha(t))-u_{j}(\alpha(t))}\notag\\[5pt]
&\le (1-\epsilon)\abs{\rho_{i}(t)-\rho_{j}(t)} +\epsilon\delta^{*}
\qquad\text{~for all~} i,j\in\cI\,,
\end{align}
and since $\zeta<\frac{1}{2}(\Delta_{\text{min}}-\delta^{*})$ by a straightforward
induction argument using \eqref{E-L4.1a} we obtain
\begin{equation}\label{E-L4.1b}
\max_{i,j\in\cI}\;\abs{\rho_{i}(t)-\rho_{j}(t)}\le
\frac{\Delta_{\text{min}}+\delta^{*}}{2}\qquad\forall t\ge0\,.
\end{equation}
For $i\in\cI$ define
\begin{equation*}
\Breve{\rho}_{i} \df\min_{\Bar{\alpha}\in\oA}\;u_{i}(\Bar{\alpha})
\qquad\text{and}\qquad
\Hat{\rho}_{i} \df\max_{\alpha\in\cA\setminus\oA}\;u_{i}(\alpha)\,,
\end{equation*}
and for $k=0,1$ define the sets
\begin{equation*}
D_{k} \df \left\{(\alpha,\rho)\in\cX : \rho_{i} \le
\frac{\Breve{\rho}_{i} + (2k+1) \Hat{\rho}_{i}}{2k+2}
+ \frac{\Delta_{\text{min}}+\delta^{*}}{4}\,,~i\in\cI\right\}\,.
\end{equation*}
Let also
\begin{equation*}
\Gamma \df \left\{(\alpha,\rho)\in\cX :
\min(\Breve{\rho}_{i}-\rho_{i},\rho_{i}-\Hat{\rho}_{i})
\ge\frac{1}{4}(\Delta_{\text{min}}-\delta^{*})\,,~i\in\cI\right\}\,,
\end{equation*}
and
\begin{equation*}
\Bar{\Gamma} \df\left\{(\alpha,\rho)\in\Gamma : \alpha \in \oA\right\}\,.
\end{equation*}
Recall the definition of $\uptau$ in \eqref{E-uptau} and in order
to simplify the notation let $\uptau_{k}\df\uptau(D_{k})$, for $k=0,1$.
Note the following:
Firstly, using \eqref{E-L4.1b}, we obtain
\begin{equation}\label{E-L4.1c}
\Gamma\subset D_{0}\setminus D_{1}\,.
\end{equation}
Secondly, since $\abs{\rho_{i}(t+1)-\rho_{i}(t)}\le \epsilon\Delta_{\text{max}}$,
we obtain
\begin{equation}\label{E-L4.1d}
\left(\uptau_{1}-\uptau_{0} - \frac{\Delta_{\text{min}}}
{4\epsilon\Delta_{\text{max}}}\right)\Ind_{\{\tau_{0}<\infty\}}\ge0
\qquad \Tilde{\Prob}_{\Bar{s}}\text{-a.s.}
\end{equation}
It is also evident that
\begin{equation}\label{E-obvious}
\Bigl\{ \limsup_{t\to\infty}\; d_{\cS}(X_{t},\cS\setminus\oS)=0\Bigr\}
\subset \{\uptau_{1}<\infty\}
\qquad \Tilde{\Prob}_{\Bar{s}}\text{-a.s.}\,,
\end{equation}
where $d_{\cS}$ is a metric in $\cS$.
It is clear from the definition of $P$
that if $x\in\Gamma$ there are two possibilities:
If a profile $\alpha\in\cA\setminus\oA$ is played, then
$\rho_{i}$ decreases in value for all $i\in\cI$, or in other words,
that $P(x,\Gamma)=1$ for all $x\in(\Gamma\cap D_{1}^{c})\setminus\Bar{\Gamma}$.
Otherwise, if a profile in $\oA$ is played, then the sample path
gets trapped in the domain of attraction of $\oS$.
This means that if $x\in\Bar{\Gamma}$ then $\Prob_{x}(\uptau_{1}<\infty)=0$,
where $\Prob_{x}$ is the probability measure induced by $P$ defined
in Section~\ref{S3}.
In this case, and by \eqref{E-L4.1c}, we also have
\begin{equation*}
P(x,\Bar{\Gamma})\ge
\min\;\left\{\frac{c}{4}(\Delta_{\text{min}}-\delta^{*}), 1-h\right\}\df \gamma
\qquad \forall x\in\Gamma\cap D_{1}^{c}\,.
\end{equation*}
Thus, using the Markov property we obtain, with
$t_{0}\df\left\lfloor\frac{\Delta_{\text{min}}}
{4\epsilon\Delta_{\text{max}}}\right\rfloor$,
\begin{equation}\label{E-L4.1f}
P^{t_{0}}(x,\Gamma\setminus\Bar\Gamma) \le (1-\gamma)^{t_{0}}
\qquad \forall x\in\Gamma\cap D_{1}^{c}\,.
\end{equation}
Conditioning on $\sF_{\uptau_{0}}$ and using
the strong Markov property, 
\eqref{E-L4.1d}, \eqref{E-L4.1f} and the foregoing, we obtain
\begin{align}\label{E-L4.1g}
\Tilde{\Prob}_{\Bar{s}}(\uptau_{1}<\infty)
&\le \Tilde{\Exp}_{\Bar{s}} \left[\Tilde{\Exp}_{\Bar{s}}
\left[\Ind_{\{\uptau_{1}<\infty\}}\mid \sF_{\uptau_{0}}\right]\right]
\notag\\[5pt]
&\le \Tilde{\Exp}_{\Bar{s}}
\left[\Prob_{X_{\uptau_{0}}}(\uptau_{1}<\infty) \right]
\notag\\[5pt]
&\le \sup_{x\in\Gamma\cap D_{1}^{c}}\;
\Prob_{x}(\uptau_{1}<\infty)\notag\\[5pt]
&\le \sup_{x\in\Gamma\cap D_{1}^{c}}\;
P^{t_{0}}(x,\Gamma\setminus\Bar\Gamma)\notag\\[5pt]
&\le \exp\left(
\left\lfloor\frac{\Delta_{\text{min}}}
{4\epsilon\Delta_{\text{max}}}\right\rfloor\log(1-\gamma)\right)\,.
\end{align}
The result then follows by \eqref{E-obvious} and \eqref{E-L4.1g}.

Next, suppose (H2) holds.
Note that in this case
$\Breve{\rho}_{i} \equiv u_{i}(\Bar{\alpha})$ for all $\Bar{\alpha}\in\oA\,.$
Pick any $\zeta < \frac{\Delta_{\text{min}}^2}{4\Delta_{\text{max}}}$.
As before we may suppose that agent $1$ trembles.
Let
$N^*(\epsilon)\triangleq\lfloor{\nicefrac{\zeta}
{\epsilon\Delta_{\text{min}}}}\rfloor$.
Let $\Breve{\uptau}$ be the first time that an action profile
in $\cA\setminus\oA$ has been played at least $N^{*}(\epsilon)$ times.
Then, at time $\Breve{\uptau}$ the aspiration level of the
initially perturbed agent $1$ satisfies:
\begin{equation*}
\rho_{1}(\Breve{\uptau}) \le \Breve{\rho}_1 + \zeta
-\epsilon\Delta_{\text{min}}N^{*}(\epsilon) \le \Breve{\rho}_{1}\,,
\end{equation*}
while the aspiration level of any agent $i\in\cI$ satisfies
\begin{equation*}
\rho_{i}(\Breve{\uptau}) \ge \Breve{\rho}_{i} - \epsilon \Delta_{\text{max}}
\left\lfloor\frac{\zeta}{\epsilon\Delta_{\text{min}}}\right\rfloor
\ge \Breve{\rho}_{i} - \epsilon \Delta_{\text{max}}
\frac{\zeta}{\epsilon\Delta_{\text{min}}} >
\Breve{\rho}_{i} - \frac{\Delta_{\text{min}}}{4}\,.
\end{equation*}
For $k=0,1$ define the sets
\begin{equation*}
\Tilde{D}_{k} \df \left\{(\alpha,\rho)\in\cX : \rho_{i} \le
\frac{\Breve{\rho}_{i} + (2k+1) \Hat{\rho}_{i}}{2k+2}\,,~i\in\cI\right\}\,,
\end{equation*}
and let $\Tilde{\uptau}_{k}\df\uptau(\Tilde{D}_{k})$, for $k=0,1$.
Also define
\begin{equation*}
\Tilde{\Gamma} \df \left\{(\alpha,\rho)\in\cX : \rho_{i} \le \Breve{\rho}_{i}
- \frac{\Delta^{2}_{\text{min}}}{4\Delta_{\text{max}}}\,,~i\in\cI\right\}\,.
\end{equation*}
It is straightforward to show that
$\Tilde{\Prob}_{\Bar{s}}(X_{\Tilde{\uptau}_{0}}\in\Tilde{\Gamma})=1$.
From this point on, we proceed as in the previous case.
\end{proof}

For the lemma that follows we need to define the following constant.
For each $\alpha^*\in\cA^*\setminus\oA$, select any $\Tilde{\alpha}\in\cA$ and
$\{j_{1},\dotsc,j_{n-1}\}\subset\cI$ which
satisfy Definition~\ref{D2.2}~(c), and define
\begin{equation*}
\Delta_{0} \df \frac{1}{2}\;\min_{\alpha^*\in\cA^*\setminus\oA}\;
\min_{ 1\le \ell \le n-1}\;
\min_{i\in\{j_{1},\dotsc,j_{\ell+1}\}}\;
\left\{u_{i}(\alpha^*)
- u_{i}\left(\Tilde{\alpha}_{j_{1}},\dotsc,\Tilde{\alpha}_{j_{\ell}},
\alpha^*_{-\{j_{1},\dotsc,j_{\ell}\}}\right)\right\}\,.
\end{equation*}
By Definition~\ref{D2.2}~(c), $\Delta_{0}>0$.

\begin{lemma} \label{L4.2}
Suppose
\begin{equation}\label{E-L4.2A}
\epsilon < \frac{\Delta_{0}\wedge\Delta_{\text{min}}}{n\Delta_{\text{max}}}\,.
\end{equation}
Then, for any strict coordination game $\mathscr{G}$ for which
$\cA^*\setminus\oA\neq\varnothing$,
there exists
a constant $M_{0}=M_{0}(h,\abs{\cA})>0$ such that
\begin{equation*}
\Hat{P}_{s^{*}\Bar{s}}\ge \frac{M_{0}}{c\,\zeta\wedge(1-h)}\qquad
\text{~for all~} s^{*}\in\cS^{*}\setminus\oS\,,~\Bar{s}\in\oS\,.
\end{equation*}
\end{lemma}

\begin{proof}
Let $s^{*}=(\alpha^{*},\rho^{*})\in\cS^{*}\setminus\oS$,
$\Bar{s}=(\Bar{\alpha},\Bar{\rho})\in\oS$.
Suppose $\Tilde{\alpha}\in\cA$ and $\{j_{1},\dotsc,j_{n-1}\}\subset\cI$
are the action profile and sequence of agents, respectively,
corresponding to $\alpha^{*}$
used in the calculation of $\Delta_{0}$.
Consider the sample paths
$s(t)=\bigl(\alpha(t),\rho(t)\bigr)$ satisfying $s(0)=s^{*}$,
$\rho_{j_{1}}(1)\in(\rho_{j_1}^*,\rho_{j_1}^*+\zeta)$, 
$\rho_{-j_{1}}(1) = \rho^{*}_{-j_{1}}$,
and
$\alpha(t)=\left(\Tilde{\alpha}_{j_{1}},\dotsc,\Tilde{\alpha}_{j_{t}},
\alpha^*_{-\{j_{1},\dotsc,j_{t}\}}\right)$, for $0<t<n$.
We have
\begin{equation}\label{E-L4.2a}
Q\bigl(s(0),s(1)\bigr)\ge\frac{1}{2n}
\frac{\bigl(c\,\zeta\wedge(1-h)\bigr)}{\abs{\cA_{j_{1}}}}\,.
\end{equation}
By \eqref{E-L4.2A}, $\rho_i^{*}-\rho_{i}(t)\le \Delta_{0}$
for all $i\in\cI$ and $t\le n$.
Therefore,
\begin{equation*}
\rho_{i}(t)-u_{i}(\alpha(t))\ge \Delta_{0}\qquad\text{~for all~}
i\in\{j_{1},\dotsc,j_{t+1}\}\,,
\end{equation*}
for $0\le t <n$ and hence we obtain
\begin{equation}\label{E-L4.2b}
P\bigl(s(t-1),s(t)\bigr) \ge h^{n-1}
\frac{\bigl(c\Delta_{0}\wedge(1-h)\bigr)}
{\abs{\cA_{j_{t+1}}}}\,,\quad 1 < t < n\,,
\end{equation}
and
\begin{equation}\label{E-L4.2c}
P\bigl(s(n-1),\Bar{s}\bigr)\ge
\frac{\bigl(c\Delta_{0}\wedge(1-h)\bigr)^{n}}
{\abs{\cA}}\,.
\end{equation}
By \eqref{E-L4.2A}, we have
\begin{equation}\label{E-L4.2d}
\Bar{\rho}_{i}-\rho_{i}(n)\ge\Delta_{\text{min}}
+\rho^{*}_{i}-\rho_{i}(n)>0\qquad\forall i\in\cI\,.
\end{equation}
By \eqref{E-L4.2d},
$\Pi\bigl(s(n-1),\Bar{s}\bigr)\ge P\bigl(s(n-1),\Bar{s}\bigr)$.
Consequently, the result follows by \eqref{E-L4.2a}--\eqref{E-L4.2c}.
\end{proof}

\subsection{Main Result}

We define inductively the following collection of sets
\begin{equation*}
\cS_{k} \df \left\{s=(\alpha,\rho)\in\bigcup_{j=0}^{k-1}
(\cS_{j})^c:
\exists i\in\cI, \alpha_{i}'\in\mathrm{BR}_{i}(\alpha)
\text{~satisfying~} \eqref{E-CG2} \text{~and~}
(\alpha_{i}',\alpha_{-i})\in\cS_{k-1}\right\}
\end{equation*}
for $\cS_{0}=\cS^*\cup\oS$.
For example, $\cS_{1}$ includes all
pure strategy states for which there exist an agent $i$ and an action
$\alpha_{i}'\in\mathrm{BR}_{i}(\alpha)$ which satisfies
\eqref{E-CG2}
(i.e., makes no other player worse off) and also
$\alpha'=(\alpha_{i}',\alpha_{-i})\in\cS_{0}$.
Let also $K$ denote the maximum $k$ for which $\cS_{k}$ is
non-empty, i.e.,
$K\df\max\;\{k\in\mathbb{N}:\cS_{k}\neq\varnothing\}\,.$
Such $K$ is well-defined since the set of action profiles $\cA$ is finite.

\begin{lemma}	\label{L4.3}
In any coordination game, the collection of sets $\{\cS_{k}\}_{k=0}^{K}$
forms a partition of $\cS$.
\end{lemma}

\begin{proof}
By definition of the collection $\{\cS_{k}\}_{k=0}^{K}$, the sets
$\cS_k$ are mutually disjoint. It remains to show that their union
coincides with $\cS$. Assume not, i.e., assume that
there exists $s\in\cS$ such that
$s=(\alpha,\rho)\notin\bigcup_{k=1}^{K}\cS_{k}$.
According to the definition of a coordination game and
Claim~\ref{CL2.1}, there
exists a sequence of action profiles $\{\alpha^{j}\}$,
such that $\alpha^{0}=\alpha$ and $\alpha^{j}=\mathrm{BR}_{i}(\alpha^{j-1})$
for some $i\in\cI$ terminates in
$\cA^*\cup\oA$.
Let $\{s^{j}\}$ denote the sequence of pure strategy states
which corresponds to $\{\alpha^{j}\}$. Then, for some $j^*$ we
have $s^{j^*}\in\cS^*\cup\oS$,
i.e., $s^{j^*}\in\cS_{0}$.
Since $s^{j^*}\in\cS_{0}$, then we should also have
that $s^{j^*-1}\in\cS_{1},\dotsc,s^{0}=s\in\cS_{j^*}$.
However, this conclusion contradicts our assumption that
$s\notin\bigcup_{k=1}^{K}\cS_{k}$.
Thus, $\bigcup_{k=1}^{K}\cS_{k}=\cS$ and therefore the
collection of sets $\{\cS_{k}\}_{k=0}^{K}$ defines a partition
for $\cS$.
\end{proof}

\begin{theorem}\label{T4.1}
Let $\mathscr{G}$ be a strict coordination game that satisfies
either one of the hypotheses (H1) or (H2) in Lemma~\ref{L4.1},
and suppose that $\zeta<C_{0}$.
Then $\pi_{s_{i}}\to0$
as $\epsilon\downarrow{0}$ for all $s_{i}\notin\oS$.
\end{theorem}

\begin{proof}
Consider the partition of $\cS$ defined by the
family of sets $\{\cS_{k}\}_{k=0}^{K}$.
Let $\Hat{P}_{\cS_{i}\cS_{j}}$ denote the sub-stochastic matrix composed
of the transition probabilities $\Hat{P}_{s_{i}s_{j}}$ for
$s_{i}\in\cS_{i}$ and $s_{j}\in\cS_{j}$.
In other words $\bigl[\Hat{P}_{\cS_{i}\cS_{j}}\bigr]$ is the block decomposition
of $\Hat{P}$ subordinate to the partition $\{\cS_{0},\cS_{1}\dotsc,\cS_{K}\}$.
Similarly, we define $\Tilde\cS^{*}\df\cS^{*}\setminus\oS$, and let
$$\begin{pmatrix}\Hat{P}_{\oS\oS} & \Hat{P}_{\oS\Tilde\cS^{*}_{\phantom{0}}}\\
\Hat{P}_{\Tilde\cS^{*}_{\phantom{0}}\oS} &
\Hat{P}_{\Tilde\cS^{*}_{\phantom{0}}\Tilde\cS^{*}_{\phantom{0}}}
\end{pmatrix}$$
denote the block decomposition of $\Hat{P}_{\cS_{0}\cS_{0}}$
subordinate to the partition $(\oS,\Tilde\cS^{*})$ of $\cS_{0}$.
From:
\begin{equation*}
\pi_{\oS} =
\pi_{\oS}\Hat{P}_{\oS\oS} +
\pi_{\oS^{c}_{\phantom{0}}}\Hat{P}_{\oS^{c}_{\phantom{0}}\oS}\,,
\end{equation*}
we obtain
\begin{equation*}
\pi_{\oS}(I-\Hat{P}_{\oS\oS})=
\pi_{\oS}\Hat{P}_{\oS\oS^{c}_{\phantom{0}}}=
\pi_{\oS^{c}_{\phantom{0}}}\Hat{P}_{\oS^{c}\oS}\,.
\end{equation*}
By Lemma~\ref{L4.1},
$\Hat{P}_{\oS\oS^{c}_{\phantom{0}}}\to0$ as $\epsilon\to0$, while by
Lemma~\ref{L4.2} for
some positive constant $\Tilde{\delta}$, which does not depend on
$\epsilon$, we have $\Hat{P}_{\Tilde\cS^{*}_{\phantom{0}}\oS}\bm{1}
\ge \Tilde{\delta}\bm{1}$.
Thus,
\begin{equation*}
\Tilde{\delta}\,\pi_{\Tilde\cS^{*}_{\phantom{0}}}\bm{1} \le
\pi_{\Tilde\cS^{*}_{\phantom{0}}}
\Hat{P}_{\Tilde\cS^{*}_{\phantom{0}}\oS}\bm{1} \le
\pi_{\oS}\Hat{P}_{\oS\oS^{c}_{\phantom{0}}}\bm{1}
=\pi_{\oS^{c}_{\phantom{0}}}\Hat{P}_{\oS^{c}_{\phantom{0}}\oS}\bm{1}
\xrightarrow[\epsilon\to0]{}0\,,
\end{equation*}
and we obtain
\begin{equation}\label{E-pi1}
\pi_{\Tilde\cS^{*}_{\phantom{0}}}\to0\quad\text{~as~} \epsilon\to0\,.
\end{equation}
Similarly, from the equation $\pi_{\cS_{0}} =
\pi_{\cS_{0}}\Hat{P}_{\cS_{0}\cS_{0}} +
\pi_{\cS^{c}_{0}}\Hat{P}_{\cS_{0}^{c}\cS_{0}}$, we obtain
$\pi_{\cS_{0}}\Hat{P}_{\cS_{0}\cS_{0}^{c}}\bm{1} =
\pi_{\cS^{c}_{0}}\Hat{P}_{\cS^{c}_{0}\cS_{0}}\bm{1}$. It is
straightforward to show, using Definition~\ref{D2.2}~(b), that for some
positive constant $\Hat{\delta}$, which does not depend on $\epsilon$,
we have $\Hat{P}_{\cS_{k}\cS_{k+1}}\bm{1} \ge \Hat{\delta}\bm{1}$ for
all $k\ge0$. Combining the equations above we get:
\begin{align*}
\Hat{\delta}\,\pi_{\cS_{0}}\bm{1}
\le \pi_{\cS_{0}}\Hat{P}_{\cS_{0}\cS_{1}}\bm{1}
&\le \pi_{\cS_{0}}\Hat{P}_{\cS_{0}\cS_{0}^{c}}\bm{1}\\[5pt]
&= \pi_{\cS^{c}_{0}}\Hat{P}_{\cS^{c}_{0}\cS_{0}}\bm{1}\\[5pt]
&=\pi_{\oS}\Hat{P}_{\oS\cS_{0}}\bm{1}
+\pi_{\Tilde\cS^{*}_{\phantom{0}}}
\Hat{P}_{\Tilde\cS^{*}_{\phantom{0}}\cS_{0}}\bm{1}
\xrightarrow[\epsilon\to0]{}0\,,
\end{align*}
where in the last line we used Lemma~\ref{L4.1} and \eqref{E-pi1}.
Thus, we have shown that $\pi_{\cS_{0}}\to0$ as $\epsilon\to0$. We
proceed by induction. Suppose $\pi_{\cS_{k}}\to0$ as $\epsilon\to0$.
Then,
\begin{equation*}
\Hat{\delta}\,\pi_{\cS_{k+1}}\bm{1} \le
\pi_{\cS_{k+1}}\Hat{P}_{\cS_{k+1}\cS_{k}}\bm{1}
\le
\pi_{\cS_{k}}\bm{1}\xrightarrow[\epsilon\to0]{}0\,,
\end{equation*}
which shows that $\pi_{\cS_{k+1}}\to0$ as $\epsilon\to0$.
By Lemma~\ref{L4.3}, the proof is complete.
\end{proof}

Theorem~\ref{T4.1} combined with
Theorem~\ref{T3.1} provides a complete
characterization of the time average asymptotic behavior of
aspiration learning in strict coordination games.

\subsection{Simulations in Network Formation Games}

In this section, we demonstrate the asymptotic behavior of
aspiration learning in coordination games as described by
Theorems~\ref{T3.1}--\ref{T4.1}.
Consider the network formation game of
Section~\ref{Sec:NFG} which, according to
Claim~\ref{CL2.2}, is a (non-strict)
coordination game. Although
Theorem~\ref{T4.1} was only shown for
strict coordination games, our intention here is to demonstrate that
it also applies to the larger class of (non-strict) coordination
games.

\begin{figure}[htbp]\centering{
\begin{minipage}{0.58\textwidth}\centering
\includegraphics[width=3.in]
{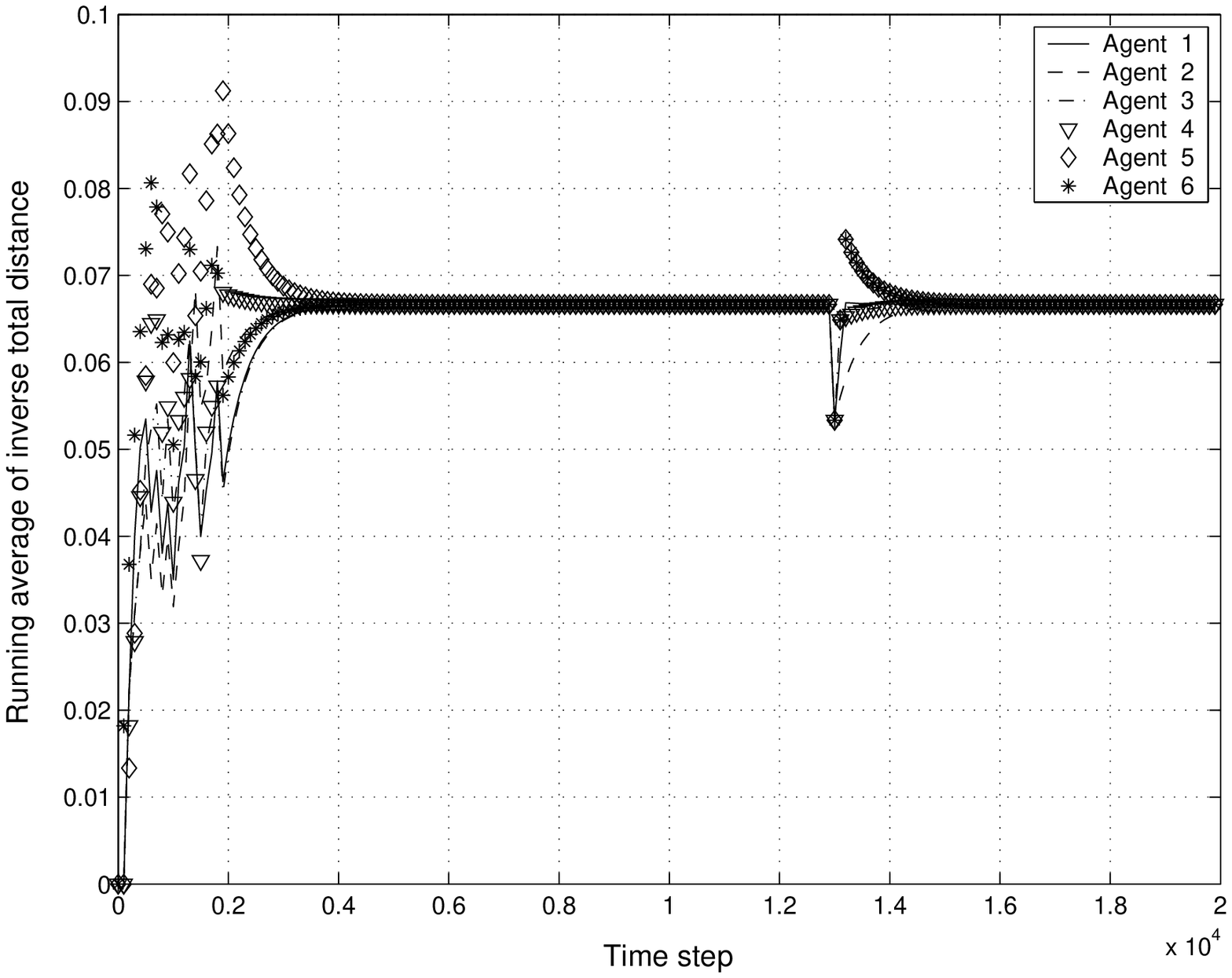}
\end{minipage}
\begin{minipage}{0.40\textwidth}
\includegraphics[width=2.in]
{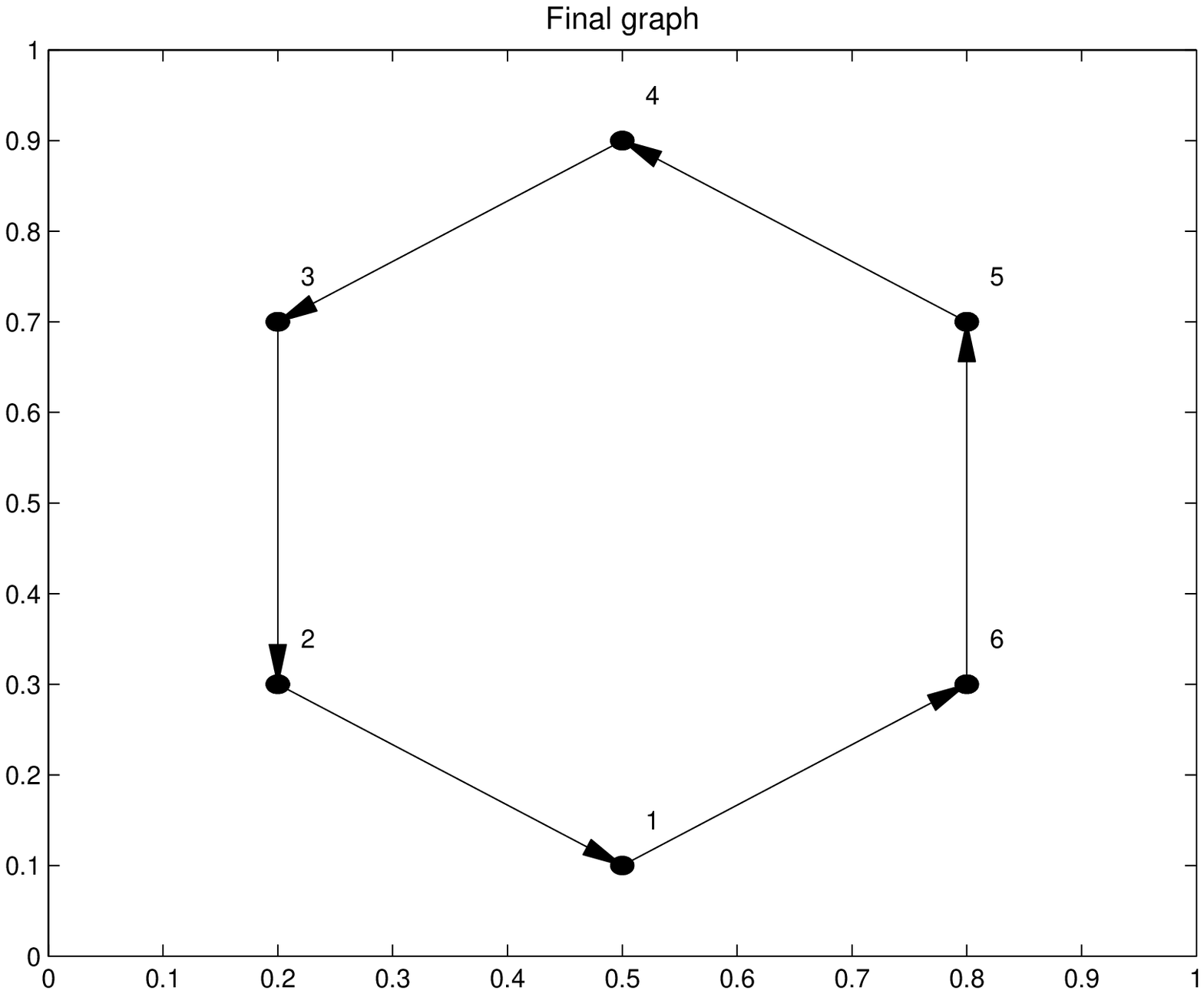}\\[10pt]
\includegraphics[width=2.in]
{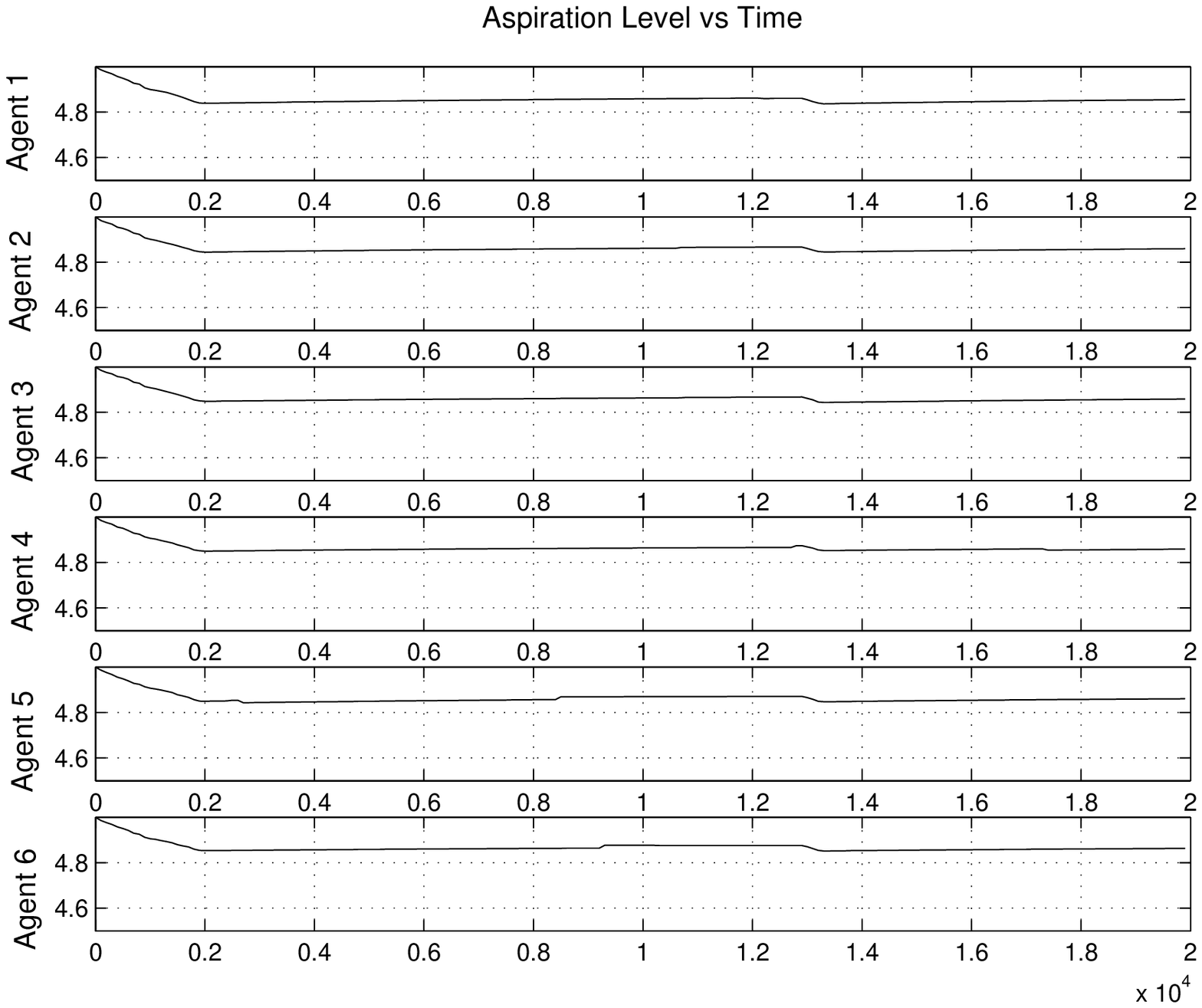}
\end{minipage}}
\caption{A typical response of aspiration learning in the network
formation game.}
\label{fig:NFGsim}
\end{figure}

We consider a set of six nodes deployed on the plane, so that the
neighbors of each node are the two immediate nodes (e.g.,
$\mathcal{N}_{1}=\{2,6\}$). Note that a payoff-dominant set of networks exists and 
corresponds to the wheel networks, where each node has a single link.
We pick the set $\oA$ of desirable 
networks as the set of wheel networks.
Note that the set $\oA$ satisfies hypothesis (H2) of Lemma~\ref{L4.1}. 

In order for the average
behavior to be observed $\lambda$ and $\epsilon$ need to be
sufficiently small. We choose: $h=0.01$, $c=0.2$, $\zeta=0.01$,
$\epsilon=\lambda=10e-4$, and $c=\nicefrac{1}{8}$. In
\figurename~\ref{fig:NFGsim}, we have plotted a typical response of
aspiration learning for this setup, where the final graph and the
aspiration level as a function of time are shown.

To illustrate better the response of aspiration learning, define the
distance from node $j$ to node $i$, denoted
$\mathrm{dist}_{G}(j,i)$, as the minimum number of hops from $j$ to
$i$. We also adopt the convention $\mathrm{dist}_G(i,i)=0$ and
$\mathrm{dist}_G(j,i)=\infty$ if there is no path from $j$ to $i$ in
$G$. The last graph in \figurename~\ref{fig:NFGsim} plots, for each
node, the running average of the inverse total distance from all
other nodes, i.e., $\nicefrac{1}{\sum_{j\in
\cI}\mathrm{dist}_{G}(j,i)}$. This number is zero if the node is
disconnected from any other node.

We observe that the payoff-dominant profile (wheel network) is
played with frequency that approaches one. In fact, the aspiration
level converges to $(n-1)-c=4.875$ and the inverse total distance
converges to $\nicefrac{1}{15}\approx 0.067$, both of which correspond
to the wheel network.

\section{Fairness in Symmetric and Coordination Games}\label{S5}

In several coordination games, establishing convergence (in the way 
defined by Theorem~\ref{T3.1}) to the set of desirable states $\oS$
(as Theorem~\ref{T4.1} showed) 
may not be sufficient. For example,
in common-pool games of Section~\ref{Sec:CPG}, convergence to $\oS$
does not guarantee that all agents get access to the common resource
in a \emph{fair} schedule. In the remainder of this section, we establish
conditions under which \emph{fairness} is also established.

\subsection{A Property of Finite Markov Chains}

In this section, we provide an approach on characterizing explicitly
the invariant distribution of a finite-state, irreducible and
aperiodic Markov chain.
We use a characterization introduced by \cite{FreidlinWentzell84}, which 
has been extensively used for showing stochastic stability arguments for several 
learning dynamics, see, e.g., \cite{Young93,Marden09}.
In particular, for finite Markov chains an invariant distribution can
be expressed as the ratio of sums of products consisting of
transition probabilities. These products can be described
conveniently by means of graphs on the set of states of the chain.

Let $\cS$ be a finite set of states, whose elements are
denoted by $s_k$, $s_{\ell}$, etc., and let a subset $\mathcal{W}$ of
$\cS$.

\begin{definition}
($\mathcal{W}$-graph) A graph consisting of arrows
$s_k\rightarrow{s_{\ell}}$
($s_k\in{\cS\setminus{\mathcal{W}}},s_{\ell}\in\cS,s_{\ell}\neq{s_k}$)
is called a $\mathcal{W}$-graph if it satisfies the following
conditions:
\begin{enumerate}
\item
every point $k\in{\cS\setminus{\mathcal{W}}}$ is the initial point of
exactly one arrow;

\item
there are no closed cycles in the graph; or, equivalently,
for any point $s_k\in{\cS\setminus{\mathcal{W}}}$ there
exists a sequence of arrows leading from it to some point
$s_{\ell}\in\mathcal{W}$.
\end{enumerate}
\end{definition}

We denote by $\mathcal{G}\{\mathcal{W}\}$ the set of
$\mathcal{W}$-graphs; we shall use the letter $g$ to denote graphs.

If $\Hat{P}_{s_ks_{\ell}}$ are nonnegative numbers, where
$s_k,s_{\ell}\in\cS$, define the product
\begin{equation*}
\varpi(g) \df \prod_{(s_k\rightarrow{s_{\ell}})\in{g}}\Hat{P}_{s_ks_{\ell}}\,.
\end{equation*}

The following Lemma holds:

\begin{lemma}[Lemma~6.3.1 in \cite{FreidlinWentzell84}]
\label{L5.1}
Let us consider a Markov chain with a finite set of states
$\cS$ and transition probabilities $\{\Hat{P}_{s_ks_{\ell}}\}$
and assume that every state can be reached from any other state in a
finite number of steps. Then the stationary distribution of the
chain is $\pi = [\pi_{s}]$, where
\begin{equation*}
\pi_{s} = \frac{R_{s}}{\sum_{s_{i}\in\cS}R_{s_{i}}}\,,
\quad s\in\cS
\end{equation*}
and $R_{s} \df \sum_{g\in{\mathcal{G}}\{s\}}\varpi(g)$.
\end{lemma}

\subsection{Fairness in Symmetric Games}

In this section, using Theorem~\ref{T3.1} and
Lemma~\ref{L5.1} we establish fairness in
\emph{symmetric} games, defined as follows:

\begin{definition}[Symmetric game]
A game $\mathscr{G}$ characterized by the action profile set $\cA$ is
symmetric if, for any two agents $i,j\in\cI$ and any action
profile $\alpha\in\cA$, the following hold: a) if
$\alpha_{i}=\alpha_{j}$, then $u_{i}(\alpha)=u_{j}(\alpha)$, and b) if
$\alpha_{i}\neq\alpha_{j}$, then there exists an action profile
$\alpha'\in\cA\setminus\{\alpha\}$, such that the following
two conditions are satisfied:
\begin{enumerate}
\item{$\alpha_{i}'=\alpha_{j}$, $\alpha_{i}=\alpha_{j}'$ and
$\alpha_k'=\alpha_k$ for all $k\neq{i,j}$;}
\item{$u_{i}(\alpha')=u_{j}(\alpha)$, $u_{i}(\alpha)=u_{j}(\alpha')$ and
$u_k(\alpha')=u_k(\alpha)$ for any $k\neq{i,j}$.}
\end{enumerate}
\end{definition}

Define the
following equivalence relation between states in $\cS$:
\begin{definition}[State equivalence]\label{D5.3}
For any two pure-strategy states $s,s'\in\cS$ such that
$s\neq s'$, let $\alpha$ and $\alpha'$ denote the corresponding action profiles.
We write $s\sim s'$ if there exist $i,j\in\cI$,
$i\neq{j}$, such that the following two conditions are satisfied:
\begin{enumerate}
\item{$\alpha_{i}'=\alpha_{j}$, $\alpha_{i}=\alpha_{j}'$ and
$\alpha_k'=\alpha_k$ for all $k\neq{i,j}$;}
\item{$u_{i}(\alpha')=u_{j}(\alpha)$, $u_{i}(\alpha)=u_{j}(\alpha')$ and
$u_k(\alpha')=u_k(\alpha)$ for any $k\neq{i,j}$.}
\end{enumerate}
\end{definition}
Since there is a one-to-one correspondence between $\cS$ and 
$\cA$, we also say that two action
profiles $\alpha$ and $\alpha'$ are equivalent, if the conditions of
Definition~\ref{D5.3} are satisfied.

\begin{lemma}\label{L5.2}
For any symmetric game and for
any two pure-strategy states $s,s'\in\cS$ such that
$s\sim{s'}$, $\pi_{s}=\pi_{s'}$.
\end{lemma}
\begin{proof}
Let us consider any two pure strategy states $s,s'\in\cS$
such that $s\sim{s'}$.
Let also consider any $\{s\}$-graph $g$, i.e., $g\in\mathcal{G}\{s\}$.
Such a graph can be identified as a collection of paths, i.e., for some
$M\geq{1}$, we have
$g =\bigcup_{m=1}^{M}g_{m}\,,$
where
\begin{equation*}
g_{m} = \bigcup_{\ell=1}^{L(m)-1}
\left(s_{\kappa_{m}(\ell)}\rightarrow s_{\kappa_{m}(\ell+1)}\right)
\end{equation*}
for some $L(m)\geq{1}$.
In the above expression, the function $\kappa_{m}$ provides an enumeration
of the states that belong to the path $g_{m}$.
Note that due to the definition of $\mathcal{G}\{s\}$-graphs, we should have that
$s_{\kappa_{m}(L(m))}=s$ for all $m=1,\dotsc,M$.
Moreover, if $M>1$, we should also have
$$\bigcap_{m=1}^{M}
\left\{s_{\kappa_{m}(1)},\dotsc,s_{\kappa_{m}(L(m)-1)}\right\}=\varnothing\,,$$
i.e., the collection of paths $\{g_{m}\}$ do not cross each other,
except at node $s$.

Let us consider any other state $s'\in\cS$ such that $s'\sim{s}$.
Since the game is symmetric, for any graph $g\in\mathcal{G}\{s\}$,
there exists a unique graph $g'\in\mathcal{G}\{s'\}$ which satisfies
$g' = \bigcup_{m=1}^{M}g_{m}'\,,$
where
\begin{equation*}
g_{m}' = \bigcup_{\ell=1}^{L(m)-1}\bigl(s_{\kappa_{m}(\ell)}'
\to s_{\kappa_{m}(\ell+1)}'\bigr)
\end{equation*}
and $s_{\kappa_{m}(\ell)}\sim s_{\kappa_{m}(\ell)}'$, $\ell=1,\dotsc,L(m)$, for
all $m\in\{1,\dotsc,M\}$.

The transition probability between any two states is a sum of
probabilities of sequences of action profiles. Since the game is symmetric,
for any such sequence of action profiles which leads, for instance, from
$s_{\kappa_{m}(\ell)}$ to $s_{\kappa_{m}(\ell+1)}$, there exists an
equivalent sequence of action profiles which leads from
$s_{\kappa_{m}(\ell)}'$ to $s_{\kappa_{m}(\ell+1)}'$. Therefore, we
should have that:
\begin{equation*}
\Hat{P}_{s_{\kappa_{m}(\ell)}s_{\kappa_{m}(\ell+1)}} =
\Hat{P}_{s_{\kappa_{m}(\ell)}'s_{\kappa_{m}(\ell+1)}'}
\end{equation*}
for any $m=1,\dotsc,M$, and hence, $\varpi(g')=\varpi(g)\,.$ In
other words, there exists an isomorphism between the graphs in the
sets $\mathcal{G}\{s\}$ and $\mathcal{G}\{s'\}$, such that any two
isomorphic graphs have the same transition probability. Thus, we
have $\pi_{s} = \pi_{s}'$ for any two states $s,s'$ such that $s\sim s'$.
\end{proof}

Lemma~\ref{L5.2} can be used to provide a more explicit
characterization of the invariant distribution $\pi$ in several
classes of coordination games which are also symmetric, e.g.,
common-pool games.

\subsection{Fairness in Common-Pool Games}

First, recall that in common-pool games we define the set of
``desirable'' or ``successful'' action profiles $\oA$ as in
\eqref{E-CPG}. To characterize more explicitly the invariant
distribution $\pi$, we define  the subset of pure-strategy states
$\oS_{i}$ that correspond to ``successful'' states for agent $i$ by
$$\oS_{i} \df \{s\in\cS: \alpha_{i} > \alpha_{j},~\forall j\ne i\}\,.$$
In other words, $\oS_{i}$ corresponds to the set of pure-strategy
states in which the action of agent $i$ is strictly larger than the
action of any other agent $j\neq{i}$.
We also define $\oS\df \bigcup_{i\in\cI}\oS_{i}$.

Note that the equivalence relation $\sim$ defines an isomorphism
among the states of any two sets $\oS_{i}$ and
$\oS_{j}$ for any $i\neq{j}$. This is due to the
fact that for any state $s_{i}\in\oS_{i}$, there
exists a unique state $s_{j}\in\oS_{j}$ such that
$s_{i}\sim s_{j}$.

\begin{lemma}\label{L5.3}
For any common-pool game,
$\pi_{\oS_{1}}=\dotsb=\pi_{\oS_{n}}.$
\end{lemma}

\begin{proof}
As already mentioned, for any $i,j\in\cI$ such that $i\neq{j}$ and
for any state $s_{i}\in\oS_{i}$, there exists a unique state
$s_{j}\in\oS_{j}$ such that $s_{j}\sim s_{i}$. Therefore, the sets
$\oS_{i}$ and $\oS_{j}$ are isomorphic with respect to the
equivalence relation $\sim$. Since a common-pool game is symmetric,
from Lemma~\ref{L5.2}, we conclude that
$\pi_{\oS_{1}}=\dotsb=\pi_{\oS_{n}}$.
\end{proof}

\begin{theorem}\label{T5.1}
Let $\mathscr{G}$ be a common-pool game which satisfies hypothesis (H1) of
Lemma~\ref{L4.1}. There exists a constant $C_0>0$ such that for any
$\zeta<C_0$, $\pi_{\oS_{i}}
\xrightarrow[\epsilon\downarrow{0}]{}\frac{1}{n},$ for all
$i\in\cI\,.$
\end{theorem}

\begin{proof}
First, recognize that the sets $\{\oS_{i}\}$ are mutually disjoint, and
$\bigcup_{i=1}^{n}\oS_{i} =\oS\,.$
Then, by Theorem~\ref{T4.1}, and for any
$\zeta<\frac{1}{2}(\Delta_{\min}-\delta^*)$, we have
$\pi_{\oS} = \sum_{i=1}^{n}\pi_{\oS_{i}} \rightarrow 1$
as $\epsilon\rightarrow{0}\,.$
Lastly, by Lemma~\ref{L5.3}, the conclusion follows.
\end{proof}

In other words, we have shown that the invariant distribution $\pi$
puts equal weight on either agent ``succeeding,'' which
establishes a form of \emph{fairness} over time.
Moreover, it puts zero
weight on states outside $\oS$ (i.e., states
which correspond to collisions) as $\epsilon\to{0}$.


\subsection{Simulations in Common-Pool Games}

Theorems~\ref{T3.1} and \ref{T5.1} provide a characterization of the
asymptotic behavior of aspiration learning in common-pool games as
$\lambda$ and $\epsilon$ approach zero. In fact, according to
Remark~\ref{Rm:Ergodicity}, the expected
percentage of time that the aspiration learning spends in any one of
the pure strategy sets $\oS_{i}$ should be equal as the perturbation
probability $\lambda\rightarrow{0}$ and $t\rightarrow\infty$ (i.e.,
fairness is established).
Moreover, the expected percentage of ``failures'' (i.e., states outside $\oS$)
approaches zero as $t\rightarrow\infty$.

We consider the following setup for aspiration learning:
$\lambda=0.001,$ $\epsilon=0.001,$ $h=0.01,$ $c=0.05,$ and $\zeta=0.05\,.$
Also, we consider a common-pool game of 2 players
and 4 actions, where $c_{0}=0, c_{1}=0.1, c_{2}=0.2$, $c_3=0.3$ and
$\tau_{0}=\tau_{1}=\tau_{2}=\tau_3=0.8$.
Note that the maximum
payoff difference within the same action profile is $\delta^*=0.1$,
and the minimum payoff difference between 
$\oA$ and $\mathcal{A}\setminus \oA$ is
$\Delta_{\min}=0.6$.
Therefore, the hypotheses of Theorem~\ref{T5.1}
are clearly satisfied since $\delta^*<\Delta_{\min}$ and $\zeta <
\frac{1}{2}(\Delta_{\min}-\delta^*)$.
Under this setup, \figurename~\ref{fig:CPGsim} demonstrates the response of
aspiration learning.
We observe, as Theorem~\ref{T5.1} predicts, that the
frequency with which either agent succeeds
approaches $\nicefrac{1}{2}$ as time increases.
Also, the frequency of collisions (i.e., the joint actions in which neither agent
succeeds) approaches zero as time increases.

\begin{figure}[htbp]
\begin{minipage}{0.58\textwidth}
\centering
 \includegraphics[width=3.in]{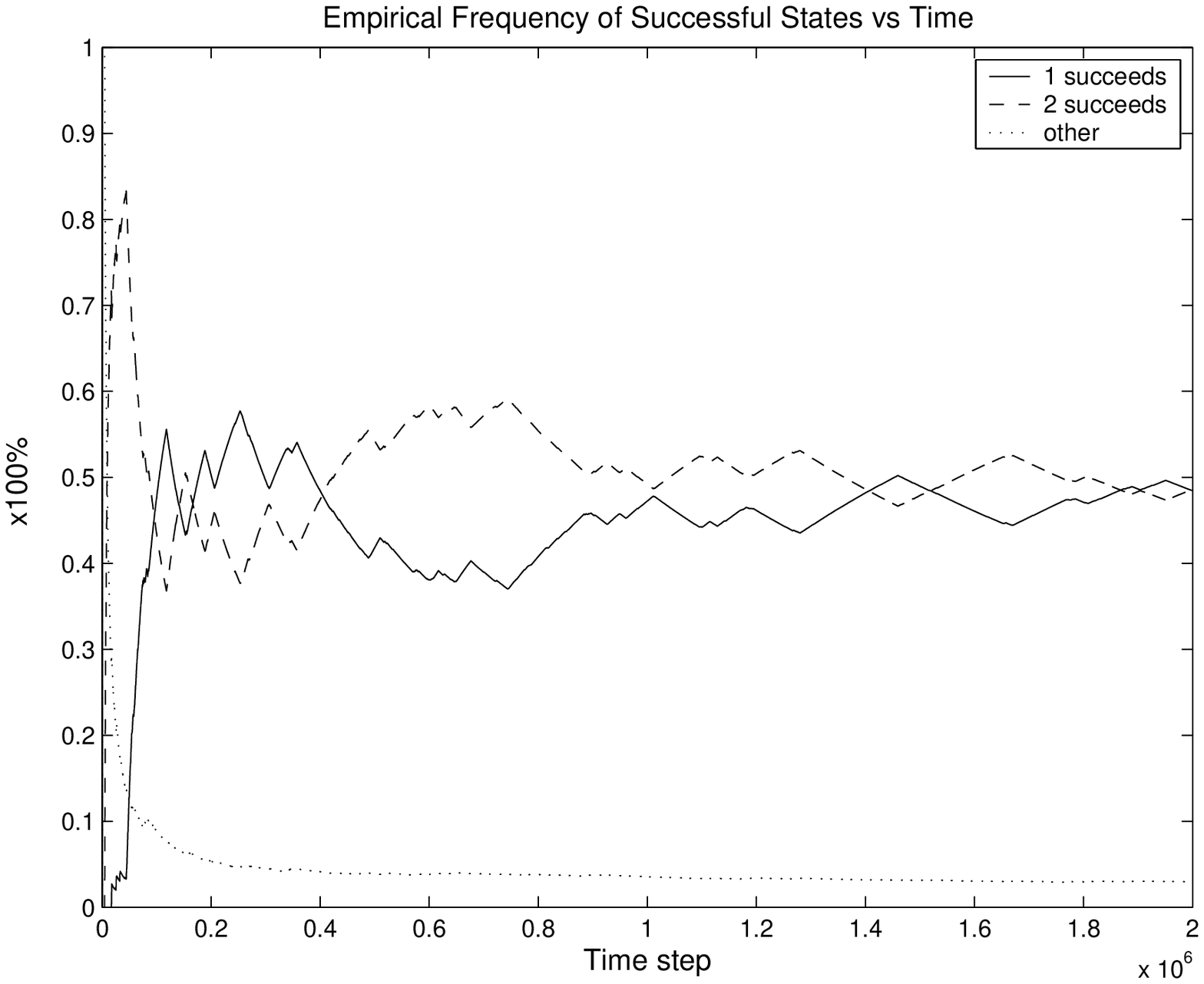}
\end{minipage}
\begin{minipage}{0.40\textwidth}
\centering
\includegraphics[width=2.in]{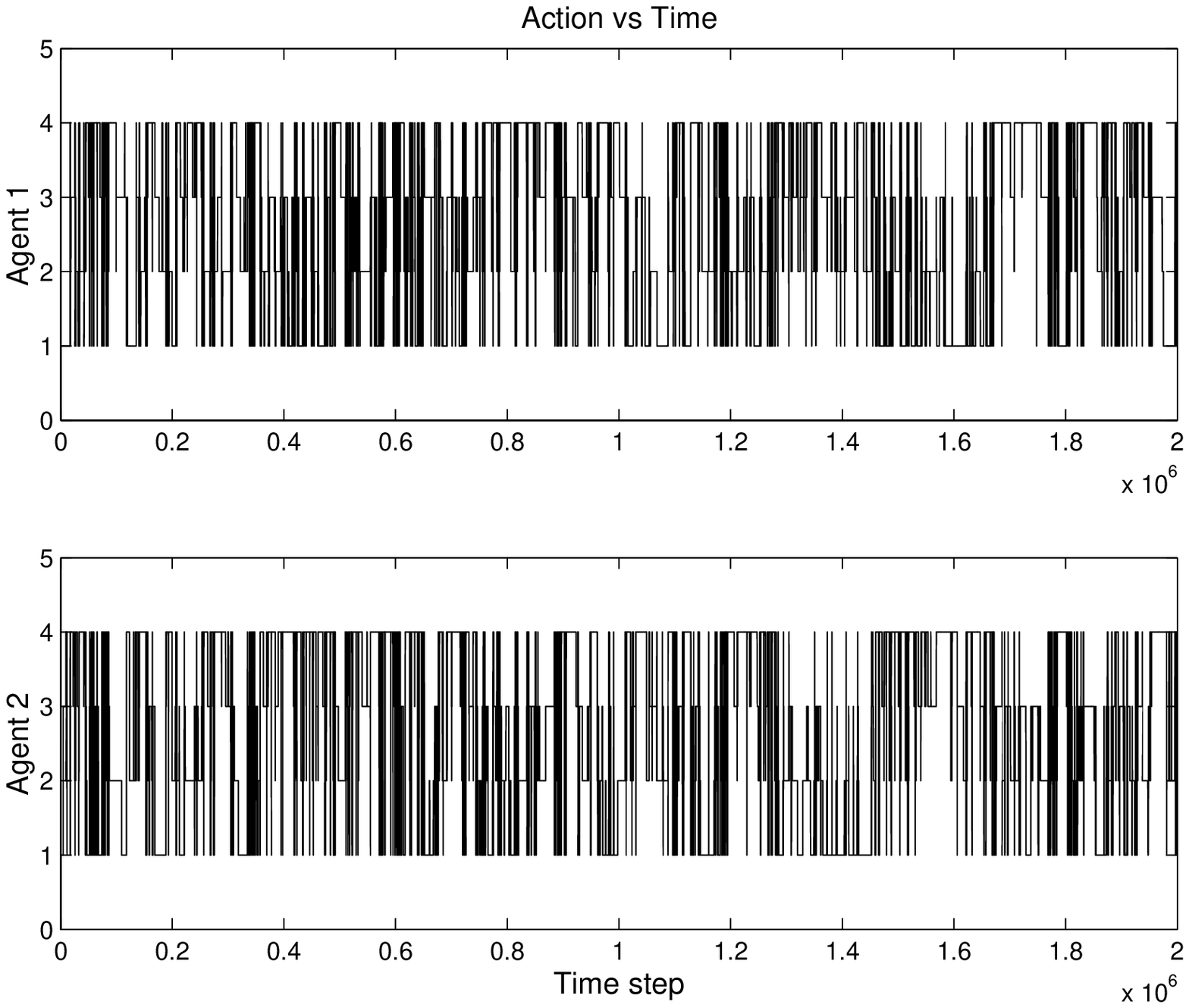}\\[10pt]
\includegraphics[width=2.in]
{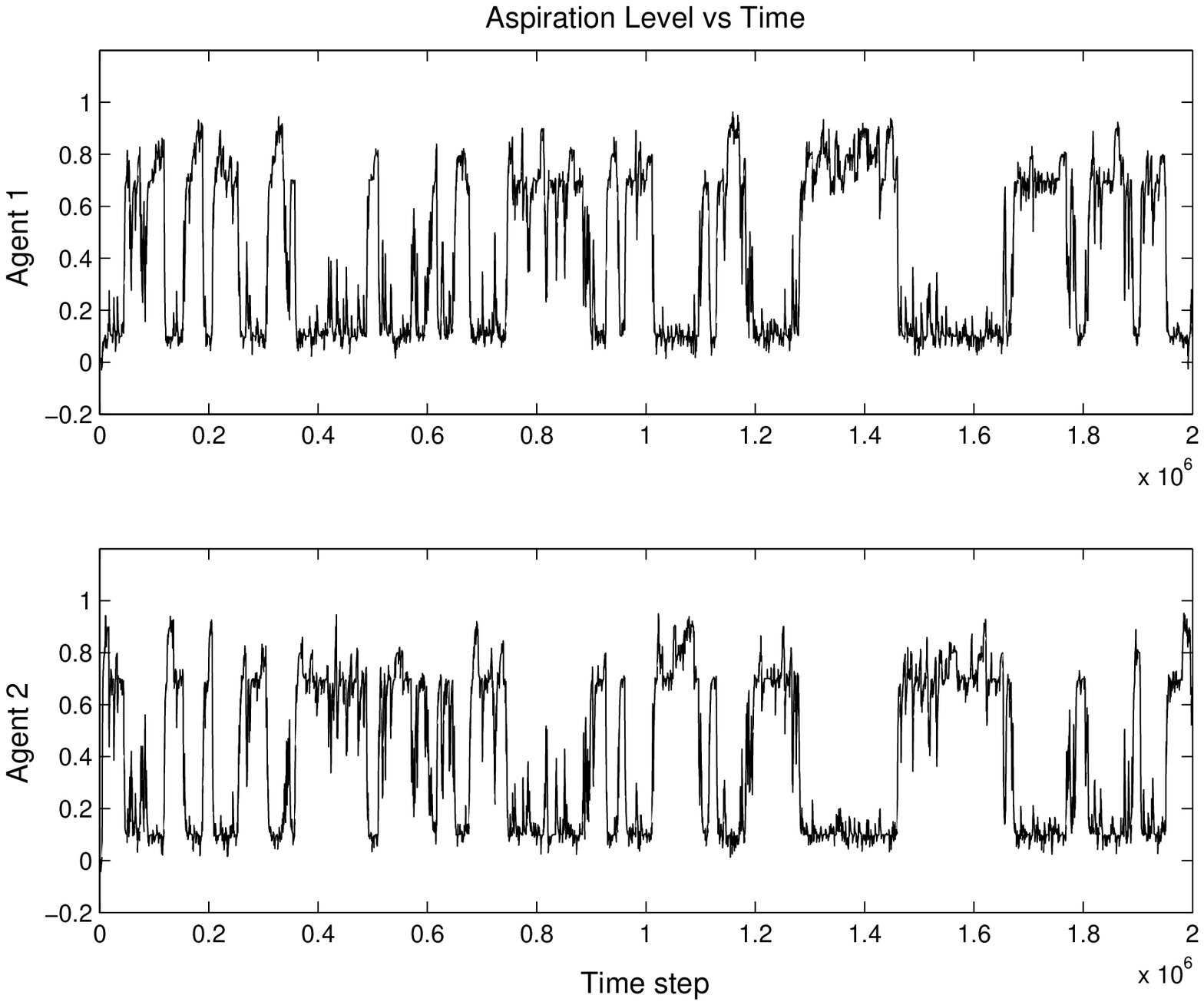}
\end{minipage}
\caption{A typical response of aspiration learning in a common-pool
game with 2 players and 4 actions.} \label{fig:CPGsim}
\end{figure}

\section{Conclusions}\label{S6}

We introduced an aspiration learning algorithm and analyzed its
asymptotic behavior in games of multiple players and actions. The
main contribution of this analysis was the establishment of a
relation between the time average behavior of the induced
infinite-state Markov chain with the invariant distribution of a
finite-state Markov chain. The establishment of this relation
allowed for characterizing the asymptotic properties of aspiration
learning when applied to generic coordination games. In particular,
we showed that over time, the efficient payoff profiles are
played (in expectation) with a frequency that can become arbitrarily
large. This analysis extended (and corrected) prior results on aspiration learning
which primarily focused on games of two players and two actions. We
further demonstrated these results through simulations on network
formation games, where distributed convergence to efficient networks
is of particular interest. Finally, we provided conditions under
which fair outcomes can be established in symmetric coordination games where
coincidence of interest among players is not so strong. For example,
we showed that in common-pool games, where multiple players compete
over utilizing a limited resource, the expected frequency at which
the common resource is exploited successfully is equally divided
among players as time increases, which establishes a form of
fairness.


\newpage
\begin{thebibliography}{10}

\bibitem{Abramson70}
{\sc N.~Abramson}, {\em The {A}loha system - another alternative for computer
  communications}, in Proc. 1970 Fall Joint Computer Conference, AFIPS Press,
  ed., 1970, pp.~281--285.

\bibitem{ArieliBabichenko11}
{\sc I.~Arieli and Y.~Babichenko}, {\em Average testing and the efficient
  boundary}, Discussion paper, Department of Economics, University of Oxford
  and Hebrew University,  (2011).

\bibitem{Bala00}
{\sc V.~Bala and S.~Goyal}, {\em A noncooperative model of network formation},
  Econometrica, 68 (2000), pp.~1181--1229.

\bibitem{ChasparisShamma_CDC08}
{\sc G.C. Chasparis and J.S. Shamma}, {\em Efficient network formation by
  distributed reinforcement}, in IEEE 47th Conference on Decision and Control,
  Cancun, Mexico, Dec. 2008, pp.~4711--4715.

\bibitem{ChoMatsui05}
{\sc I.~K. Cho and A.~Matsui}, {\em Learning aspiration in repeated games},
  Journal of Economic Theory, 124 (2005), pp.~171--201.

\bibitem{Chun04}
{\sc B.~G. Chun, R.~Fonseca, I.~Stoica, and J.~Kubiatowicz}, {\em
  Characterizing selfishly constructed overlay routing networks}, in Proc. of
  IEEE INFOCOM 04, Hong-Kong, 2004.

\bibitem{FreidlinWentzell84}
{\sc M.~I. Freidlin and A.~D. Wentzell}, {\em Random perturbations of dynamical
  systems}, Springer-Verlag, New York, NY, 1984.

\bibitem{FudenbergLevine98}
{\sc D.~Fudenberg and D.~K. Levine}, {\em The Theory of Learning in Games}, MIT
  Press, Cambridge, MA, 1998.

\bibitem{Han08}
{\sc Z.~Han and K.J.~Ray Liu}, {\em Resource Allocation for Wireless Networks},
  Cambridge University Press, 2008.

\bibitem{Lerma03}
{\sc O.~Hernandez-Lerma and J.~B. Lasserre}, {\em Markov Chains and Invariant
  Probabilities}, Birkhauser Verlag, 2003.

\bibitem{Inaltekin05}
{\sc H.~Inaltekin and S.~Wicker}, {\em A one-shot random access game for
  wireless networks}, in International Conference on Wireless Networks,
  Communications and Mobile Computing, 2005.

\bibitem{JacksonWolinsky96}
{\sc M.~Jackson and A.~Wolinsky}, {\em A strategic model of social and economic
  networks}, Journal of Economic Theory, 71 (1996), pp.~44--74.

\bibitem{Karandikar98}
{\sc R.~Karandikar, D.~Mookherjee, and D.~Ray}, {\em Evolving aspirations and
  cooperation}, Journal of Economic Theory, 80 (1998), pp.~292--331.

\bibitem{Kim99}
{\sc Y.~Kim}, {\em Satisficing and optimality in $2\times{2}$ common interest
  games}, Economic Theory, 13 (1999), pp.~365--375.

\bibitem{Komali08}
{\sc R.~Komali, A.~B. MacKenzie, and R.~P. Gilles}, {\em Effect of selfish node
  behavior on efficient topology design}, IEEE Transactions on Mobile
  Computing, 7 (2008), pp.~1057--1070.

\bibitem{Lewis02}
{\sc D.~Lewis}, {\em Convention: A Philosophical Study}, Blackwell Publishing,
  2002.

\bibitem{MardenYoungPao11}
{\sc J.~Marden, H.~P. Young, and L.~Y. Pao}, {\em Achieving {P}areto optimality
  through distributed learning}, Discussion paper, Department of Economics,
  University of Oxford,  (2011).

\bibitem{Marden09}
{\sc J.~R. Marden, H.~P. Young, G.~Arslan, and J.~S. Shamma}, {\em Payoff-based
  dynamics for multi-player weakly acyclic games}, SIAM Journal on Control and
  Optimization, 48 (2009), pp.~373--396.

\bibitem{Meinhardt99}
{\sc H.~Meinhardt}, {\em Common pool games are convex games}, Journal of Public
  Economic Theory, 1 (1999), pp.~247--270.

\bibitem{Osborne94}
{\sc M.~J. Osborne and A.~Rubinstein}, {\em A Course in Game Theory}, MIT
  Press, Cambridge, MA, 1994.

\bibitem{Pazgal95}
{\sc A.~Pazgal}, {\em Satisficing leads to cooperation in mutual interest
  games}, Int J Game Theory, 26 (1997), pp.~698--712.

\bibitem{Posch99}
{\sc M.~Posch, A.~Pichler, and K.~Sigmund}, {\em The efficiency of adapting
  aspiration levels}, Biological Sciences, 266 (1999), pp.~1427--1435.

\bibitem{Sandholm10}
{\sc William~H. Sandholm}, {\em Population Games and Evolutionary Dynamics},
  The MIT Press, Cambridge, MA, 2010.

\bibitem{Santi05}
{\sc P.~Santi}, {\em Topology Control in Wireless Ad Hoc and Sensor Networks},
  Wiley, 2005.

\bibitem{Simon55}
{\sc H.~A. Simon}, {\em A behavioural model of rational choice}, Quarterly
  Journal of Economics, 69 (1955), pp.~99--118.

\bibitem{Tembine09}
{\sc H.~Tembine, E.~Altman, R.~ElAzouri, and Y.~Hayel}, {\em Correlated
  evolutionary stable strategies in random medium access control}, in
  International Conference on Game Theory for Networks, 2009, pp.~212--221.

\bibitem{Vanderschraaf01}
{\sc P.~Vanderschraaf}, {\em Learning and Coordination}, Routledge, New York,
  NY, 2001.

\bibitem{Young93}
{\sc H.~P. Young}, {\em The evolution of conventions}, Econometrica, 61 (1993),
  pp.~57--84.

\bibitem{Young04}
\leavevmode\vrule height 2pt depth -1.6pt width 23pt, {\em Strategic Learning
  and Its Limits}, Oxford University Press, New York, NY, 2004.

\end{thebibliography}

\end{document}